\documentclass[a4paper,10pt]{article}

\usepackage[utf8]{inputenc}
\usepackage{authblk}
\usepackage[english]{babel}
\usepackage{amsmath}
\usepackage{amsfonts}
\usepackage{amssymb}
\usepackage{amsthm}
\usepackage{mathrsfs} 
\usepackage{verbatim} 
\usepackage{etoolbox} 
\usepackage[stretch=5,shrink=5]{microtype} 
\usepackage{hyperref} 
\usepackage{csquotes} 
\usepackage{mathtools} 
\usepackage[justification=centering,font={small}]{caption}
\usepackage{array} 
\usepackage{multirow} 
\usepackage{color}
\usepackage{float}
\usepackage{setspace}
\usepackage{wrapfig}
\usepackage{forest}

\newcolumntype{M}[1]{>{\centering\arraybackslash}m{#1}}
\newcolumntype{N}{@{}m{0pt}@{}}

\usepackage[backend=bibtex, doi=false, isbn=false, url=false, maxbibnames=99]{biblatex} 

\AtEveryBibitem{\clearlist{language}} 

\newtheorem{assumption}{Assumption} 
\newtheorem{definition}{Definition}
\newtheorem{theorem}{Theorem}  
\newtheorem{proposition}{Proposition} 
\newtheorem{remark}{Remark} 
\newtheorem{lemma}{Lemma} %

\makeatletter

\newcommand*\wthelper[2]{%
        \hbox{\dimen@\accentfontxheight#1%
                \accentfontxheight#11.2\dimen@ 
                $\m@th#1\widetilde{#2}$%
                \accentfontxheight#1\dimen@
        }%
}
\newcommand*\accentfontxheight[1]{%
        \fontdimen5\ifx#1\displaystyle
                \textfont
        \else\ifx#1\textstyle
                \textfont
        \else\ifx#1\scriptstyle
                \scriptfont
        \else
                \scriptscriptfont
        \fi\fi\fi3
}
\makeatother
\newcommand{\powerset}[1]{\mathcal{P} \left( #1 \right) }
\newcommand{\G}{\mathbb{G}}
\newcommand{\GT}{\widetilde{\G}} 
\newcommand{\N}{\mathbb{N}}
\newcommand{\Z}{\mathbb{Z}}

\newcommand{\bigO}[1]{\mathcal{O}\left( #1 \right)}
\newcommand{\prob}[1]{\mathbb{P}\left[#1\right]}
\newcommand{\adv}{\textbf{Adv}}
\newcommand{\tree}{\mathfrak{T}}
\newcommand{\ith}{\text{-th}}
\newcommand{\varRow}{l}
\newcommand{\varColumn}{m}
\newcommand{\pvec}[2][]{{\hspace{-0.05em}\vec{\hspace{0.05em}#2}\mkern2mu\vphantom{#2}}^\prime_{\hspace{-0.18em}#1}}
\newcommand{\ppvec}[2][]{{\hspace{-0.05em}\vec{\hspace{0.05em}#2}\mkern2mu\vphantom{#2}}^{\prime\prime}_{\hspace{-0.18em}#1}}
\newcommand{\definline}[1]{(\textit{#1})}
\newcommand{\matrixset}[3]{\text{M}_{#1\times #2}(#3)}
\newcommand{\maps}[2]{$ \left( #1 \right) \mapsto \left( #2 \right) $}
\newcommand{\mapssingleoutput}[2]{$ \left( #1 \right) \mapsto #2 $}
\newcommand{\singlefunction}[1]{\texttt{#1}}
\newcommand{\algorithm}[2]{\texttt{#1.#2}}
\newcommand{\algorithmdef}[4]{\item[#1.#2 \maps{#3}{#4}:]}
\newcommand{\algorithmdefsingleoutput}[4]{\item[#1.#2 \mapssingleoutput{#3}{#4}:] }
\newcommand{\singlefunctiondef}[3]{\item[ \textbf{#1} \maps{#2}{#3}.\hspace{0.1em} ]}
\newcommand{\singlefunctiondefsingleoutput}[3]{\item[ \textbf{#1} \mapssingleoutput{#2}{#3}.\hspace{0.1em} ]}
\newcommand{\randomchoose}[1]{We choose randomly the element\ifstrequal{#1}{s}{}{s} 
}

\newcommand{\game}[2]{$\mathbf{#1_{#2}}$}
\newcommand{\mathgame}[2]{\mathbf{#1_{#2}}}
\newcommand{\phase}[1]{%
\textbf{#1 phase.} \hspace{0pt}}
\newcommand{\gamedef}[3]{
\textbf{Definition of \game{#1}{#2}.}
#3}
\newcommand{\EK}{\mathrm{EK}}
\newcommand{\CH}{\mathrm{CH}}
\newcommand{\CV}{\mathrm{CV}}
\newcommand{\PV}{\mathrm{PV}}
\newcommand{\PI}{\mathrm{PI}}
\newcommand{\normal}{S}
\newcommand{\semifunctional}{SF}
\newcommand{\sfone}{SF1}
\newcommand{\sftwo}{SF2}
\newcommand{\sfthree}{SF3}
\newcommand{\sffour}{SF4}
\newcommand{\random}{R}
\newcommand{\typeone}{type~1}
\newcommand{\typetwo}{type~2}
\newcommand{\disclaimerSUE}{The SUE keys are generated in a similar fashion, for which we refer to the paper by Lee et al.~\cite{lee2013RSABE}.

}

\definecolor{mildgray}{RGB}{220,220,220} 
\definecolor{darkgray}{RGB}{150,150,150} 
\colorlet{Rnode}{darkgray}
\definecolor{mildgray}{RGB}{230,230,230} 
	\definecolor{Rnode}{RGB}{170,170,170} 
	\definecolor{consideredNode}{RGB}{230,230,230}
	\definecolor{consideredEdge}{RGB}{0,0,0}
	\definecolor{boxPV}{RGB}{0,0,0}
	\newlength{\consideredWidth}
	\newlength{\boxPVWidth}
	\definecolor{revokedNode}{RGB}{170,170,170}
	\definecolor{revokedEdge}{RGB}{0,0,0}
	\definecolor{boxCV}{RGB}{0,0,0}	
	\definecolor{consideredNode}{RGB}{230,230,230}
	\definecolor{consideredEdge}{RGB}{0,0,0}
	\definecolor{revokedNode}{RGB}{170,170,170}
	\definecolor{revokedEdge}{RGB}{0,0,0}
	\definecolor{boxSet}{RGB}{0,0,0}
	\newlength{\boxSetWidth}
	\newlength{\revokedWidth}
	\newlength{\boxCVWidth}
\renewcommand{\emptyset}{\varnothing}

\renewcommand{\setminus}{\smallsetminus}

\title{A proof of security for a key-policy \texorpdfstring{\\}{}RS-ABE scheme}
\author[1]{Federico Giacon\thanks{federico.giacon@rub.de}}
\author[2]{Riccardo Aragona\thanks{riccardo.aragona@unitn.it}}
\author[2]{Massimiliano Sala\thanks{maxsalacodes@gmail.com }}
\affil[1]{Horst G\"ortz Institute for IT-Security, Ruhr-University Bochum, Germany}
\affil[2]{Department of Mathematics, University of Trento, Italy}
\date{}

\begin{document}
	\maketitle
	\begin{abstract}
		A revocable-storage attribute-based encryption  (RS-ABE) scheme is an encryption scheme which extends attribute-based encryption by introducing user revocation. A key-policy RS-ABE scheme links each key to an access structure. We propose a new key-policy RS-ABE scheme whose security we prove in term of indistinguishability under a chosen-plaintext attack (IND-CPA).
	\end{abstract}

	\section{Introduction}
	\label{intro}
	
	
	Attribute-based encryption (ABE) is a flavor of public-key encryption which has seen a growing interest in the last years. 
	An ABE scheme builds its ciphertexts and its keys employing \emph{attributes}. Attributes are used to manage the access to a certain document using a fixed key. In a \emph{key-policy} ABE scheme each ciphertext links the encrypted plaintext to a set of attributes, and each key is linked to an \emph{access structure}. An access structure is a list of sets of attributes which the key is able to decrypt. Thus a key can decrypt a ciphertext if and only if the set of attributes of the ciphertext belongs to the access structure of the key.
	In a similar way, a \emph{ciphertext-policy} reverses the roles of set of attributes and access structure: in this case, a key is associated with the set of attributes, and a ciphertext with an access structure.
	The choice between key-policy and ciphertext-policy depends on the final application of the scheme. For example, a digital television platform could favour key policy to manage their TV licenses: associating attributes to encrypted programs, the clients are able to decrypt only the programs which belong to their license. Conversely, a bank could prefer a ciphertext policy to manage the work of its personnel: if we associate attributes to the cryptographic keys of each employee, then each member is able to decrypt only the documents within their competence, i.e., those corresponding to attributes associated with their keys.
	Different ABE constructions have been proposed over time, with gradually increased security, reliability and efficiency.
	For this reason the concept of ABE has been extended in order to widen its use-cases.
	Our interest lies specifically on user revocation and in this paper we design a public-key scheme which enjoys ABE for encryption, while allowing to  revoke users arbitrarily.
	This is the goal of revocable-storage attribute-based encryption (RS-ABE), a scheme described by \citeauthor{lee2013RSABE}~\cite{lee2013RSABE} employing a ciphertext policy. To be more precise, our aim is to build a key-policy version of such scheme and to prove its security in a theoretical framework fit to the application context. Starting from the ideas in \cite{lee2013RSABE}, we modify concepts and techniques described therein to adapt them to our goals.\\
	We observe that a key-policy RS-ABE scheme was independently described by \citeauthor{lee2014}~\cite{lee2014}, but we employ a different construction and different security assumptions, reaching an independent security result.

	In Section~\ref{sec:structure} we describe the RS-ABE framework, and we give the definition of CPA-security. In Section~\ref{assumptions.sec} we describe the assumptions required for the security of our scheme. In Section~\ref{sec:block} we give an overview of the building blocks of our scheme: Complete Subset (CS), Self-Updatable Encryption (SUE) and key-policy Attribute Based Encryption. In Section~\ref{sec:scheme} we are finally able to describe our scheme in detail and  state our main result on its security.  Section~\ref{sec:proof} is entirely devoted to the proof of our claimed theorem. Finally, in Section~\ref{efficiency.sec} we discuss the efficiency of our key-policy RS-ABE scheme and draw some conclusions, sketching future work.
	

	\section{Structure}\label{sec:structure}
	We provide a high-level description of a \emph{key-policy revocable-storage at\-tri\-bute-based encryption} scheme. $\mathfrak{A}$ is the set of all possible attributes and an access structure $\mathbb{A}$ is a set of subsets of $\mathfrak{A}$ (for a formal definition, see Definition~\ref{def.accessstructure}). Moreover we denote with $\mathfrak{U}$ the set of all users. 
	Starting from $\mathfrak{A}$ and other public parameters, an authority  $\mathcal{C}$ must create some \emph{public information} ($\PI$) and two general keys:  the \emph{general public key} and the \emph{master key}. The public information contains the general setting and is known to everyone. The general public key, which we shorten to \emph{public key} ($\mathrm{PK}$), can be used by any user to encrypt, even by anyone having access to the system and knowing only $\PI$ and $\mathrm{PK}$. The master key ($\mathrm{MK}$) is used only by 
	$\mathcal{C}$ to create the users' private keys (no user has a personal public key), when user requests it, and the (general) \emph{time-update key} ($\mathrm{TK}$), at any time update. The latter is known to everyone and incorporates the information on the updated list of revoked users. In order to decrypt, any user needs her own \emph{private key} ($\mathrm{SK}$), $\mathrm{TK}$ and $\PI$.\\
	This scheme is described by seven Probabilistic Polynomial Time (PPT) algorithms. 
		\begin{description}
		\singlefunctiondef{Setup}{\lambda,\mathfrak{A},T_{\textnormal{max}},N_{\textnormal{max}}}{\mathit{\mathrm{MK}},\PI,\mathit{\mathrm{PK}}}
		\label{RS-ABE_Setup-theory}
			$\lambda$ is the security parameter,
			$\mathfrak{A}$ is the set of all possible attributes,
			$T_{\textnormal{max}}$ is the maximum time which can be used inside the system,
			$N_\textnormal{max}$ is the total number of users who can receive a key. W.l.o.g. we can assume $N_\textnormal{max}=2^d$ with $d\geq 1$.
			The outputs are the master (private) key,   the public information and the (general) public key.
		
		\singlefunctiondefsingleoutput{GenKey}{\mathit{\mathrm{PI}},\mathit{\mathrm{PK}},\mathit{\mathrm{MK}},\mathbb{A},u}{\mathit{\mathrm{SK}}_{\mathbb{A},u}}
		The inputs are
		\begin{itemize} 
		\item all three outputs of \singlefunction{Setup}, that is, $\mathit{\mathrm{PI}}$, $\mathit{\mathrm{PK}}$, and $\mathit{\mathrm{MK}}$, 
		\item $\mathbb{A}$, an access structure listing the sets of attributes which the key is able to decrypt, and
		\item $u$, the user assignee of the key.
		\end{itemize}
		The output is 
		$\mathit{\mathrm{SK}}_{\mathbb{A},u}$, the private key for the user $u$.

		\singlefunctiondefsingleoutput{UpdateKey}{\mathit{\mathrm{PI}},\mathit{\mathrm{PK}},\mathit{\mathrm{MK}},T,R}{\mathit{\mathrm{TK}}_{T,R}}
		This function generates a key which links the time $T$ with a set of \emph{revoked users} $R$. The remaining inputs are
		$\mathit{\mathrm{PI}}$, $\mathit{\mathrm{PK}}$, and
		$\mathit{\mathrm{MK}}$.
		The output is
		$\mathit{\mathrm{TK}}_{T,R}$, the \emph{time-update key} for the time $T$ (and with revoked users $R$).
		
		\singlefunctiondefsingleoutput{Encrypt}{\mathit{\mathrm{PI}},\mathit{\mathrm{PK}},M,S,T}{\mathrm{CT}_{S,T}}
		The inputs are $\mathit{\mathrm{PI}}$, $\mathit{\mathrm{PK}}$, a plaintext
		$M$, 
		$S\subseteq \mathfrak{A}$ (the set of attributes required to decrypt the plaintext) and
		the encryption time $T$.
		The output is
		$\mathrm{CT}_{S,T}$, the ciphertext.
		
		\singlefunctiondefsingleoutput{Decrypt}{\mathit{\mathrm{PI}},\mathrm{CT}_{S,T},\mathit{\mathrm{SK}}_{\mathbb{A},u},\mathit{\mathrm{TK}}_{T^{\prime}\!,R}}{M}
		This is the function merging together both the secret key $\mathit{\mathrm{SK}}_{\mathbb{A},u}$ and the time-update key $\mathit{\mathrm{TK}}_{T^{\prime}\!,R}$ to decrypt a ciphertext $\mathrm{CT}_{S,T}$. Moreover, the public information $\mathit{\mathrm{PI}}$ is used.
		The output is
		$M$, the decrypted plaintext.
		
		\singlefunctiondefsingleoutput{Update$\mathrm{CT}$}{\mathit{\mathrm{PI}},\mathit{\mathrm{PK}},\mathrm{CT}_{S,T}}{\mathrm{CT}_{S,T+1}}
		The goal of this function is to update the time $T$ of a ciphertext $\mathrm{CT}_{S,T}$ by a unit, using only $\mathit{\mathrm{PI}}$ and $\mathit{\mathrm{PK}}$.		
		The output is
		$\mathrm{CT}_{S,T+1}$, a ciphertext for the same plaintext of $\mathrm{CT}_{S,T}$, but with updated time.
		
	\end{description}
	
	W.l.o.g. we will always assume that the set $\mathfrak{U}$ of users is the largest possible, that is, $|\mathfrak{U}|=N_{\textnormal{max}}=2^d$.
	
	The \emph{correctness} of a scheme describes which keys can decrypt a ciphertext. This scheme is correct \underline{if}, fixing any $\mathit{\mathrm{PI}}$, $\mathit{\mathrm{PK}}$, and $\mathit{\mathrm{MK}}$ obtained as output of \singlefunction{Setup}, for any plaintext $M$ and any ciphertext $\mathrm{CT}_{S,T}$ output of \singlefunction{Encrypt}, for any key $\mathit{\mathrm{SK}}_{\mathbb{A},u}$ output of \singlefunction{GenKey}, and any key $\mathit{\mathrm{TK}}_{T^{\prime},R}$ output of \singlefunction{UpdateKey}, \underline{then} the output of \singlefunction{Decrypt}$\left({\mathit{\mathrm{PI}},\mathrm{CT}_{S,T},\mathit{\mathrm{SK}}_{\mathbb{A},u},\mathit{\mathrm{TK}}_{T^{\prime}\!,R}}\right)$ is exactly $M$ \underline{when} all the following are true:
	\begin{itemize}
		\item
		$u\not\in R$, i.e., the user $u$ is not revoked;
		\item
		$T\leq T^{\prime}$, i.e., the ciphertext is older than the update key;
		\item
		$S\in\mathbb{A}$, i.e., the set of attributes of the ciphertext is authorized by the key.
	\end{itemize}

	We also require that, for any ciphertext $\mathrm{CT}_{S,T}$ output of \singlefunction{Encrypt}$\left({\mathit{\mathrm{PK}},M,S,T}\right)$, decrypting the output of \singlefunction{Update$\mathrm{CT}$}$({\mathit{\mathrm{PK}},\mathrm{CT}_{S,T}})$ yields the same as decrypting the output of \singlefunction{Encrypt}$(\mathit{\mathrm{PK}},M,S,T+1)$. 

	\subsection{Definition of security}
	\label{security.sec}
	The security of our scheme is described in term of \emph{indistinguishability under a chosen-plaintext attack} (\emph{IND-CPA}).
	An informal description follows.\\
	An adversary is allowed to query the challenger for a polynomial number of private keys and time-update keys. When satisfied, the adversary presents two different plaintexts $M_*^0$ and $M_*^1$; one of them is randomly chosen by the challenger and given back encrypted ($\mathrm{CT}_*$) to the adversary. The adversary may then continue querying the challenger for other private keys and time-update keys, and must eventually decide which plaintext between $M_*^0$ and $M_*^1$ corresponds to $\mathrm{CT}_*$. Informally, we may consider an RS-ABE scheme secure if the adversary cannot reliably guess which of the two messages was encrypted.\\
	We now provide the corresponding formal security game.
	\begin{definition} 
		\label{security.def}
		The security game describing IND-CPA is played by a PPT  adversary $\mathcal{A}$ and a challenger $\mathcal{C}$, and consists of five distinct phases. At the start we fix the attributes $\mathfrak{A}$, the maximum time $T_{\textnormal{max}}$, the number of users $N_{\textnormal{max}}$, and the security parameter $\lambda$. Each of these parameters are known by both~$\mathcal{A}$ and $\mathcal{C}$.
		\begin{description}
			\item[Setup]
			$\mathcal{C}$ runs \singlefunction{Setup}$(\lambda,\mathfrak{A},T_{\textnormal{max}},N_{\textnormal{max}})$ to obtain the master secret key $\mathit{\mathrm{MK}}$, the public information $\mathit{\mathrm{PI}}$ and the public key $\mathit{\mathrm{PK}}$. The former is kept secret by $\mathcal{C}$, while the others are given to~$\mathcal{A}$.
			
			\item[Query, I]
			$\mathcal{A}$ is now allowed to query $\mathcal{C}$. The maximum number of allowed queries $n_q^{I}$ is polynomial with respect to $\lambda$.
			The queries\footnote{The querying process is adaptive, i.e., the choice for the type of query and the input for the next query may depend on the output of the previous queries. } may be of two different types: private-key queries,  indexed with $i$, or time-update queries, indexed with $j$. 
			\begin{enumerate}
				\item
				(private-key query)
				$\mathcal{A}$ chooses a pair $(\mathbb{A}_i,u_i)$ consisting of an access \linebreak{} 
				 structure $\mathbb{A}_i$ and a user $u_i$ who was not previously requested. 
				$\mathcal{C}$ runs \singlefunction{GenKey}$(\mathit{\mathrm{PI}},\mathit{\mathrm{PK}},\mathit{\mathrm{MK}},\mathbb{A}_i,u_i)$ and gives the output $\mathit{\mathrm{SK}}_{\mathbb{A}_i,u_i}$ to $\mathcal{A}$.
				\item
				(time-update query)
				$\mathcal{A}$ selects a pair $(T_j,R_j)$, where $R_j$ is an arbitrary set of users and $T_j$ is a time for which there was no previous request for a private key (${\forall k < j,}\: {T_j\neq T_{k}}$).
				$\mathcal{C}$ runs \singlefunction{UpdateKey}$(\mathit{\mathrm{PI}},\mathit{\mathrm{PK}},\mathit{\mathrm{MK}},T_j,R_j)$ and hands its output 
				$\mathit{\mathrm{TK}}_{T_j,R_j}$ 
				to~$\mathcal{A}$.
			\end{enumerate}

			\item[Challenge]
			\label{security.challenge}
			$\mathcal{A}$ chooses two plaintexts $M_*^0$ and $M_*^1$, together with a time $T_*$ and a set of attributes $S_*$.
			We require also that, for each previous private-key query $i$ and  time-update query $j$ occurring in Query~I, the following condition holds\footnote{Since $T_*$ and $S_*$ are going to be used in the encryption,   $\mathcal{A}$ must not be able to decrypt the message using the keys she previously queried.}
				\[(S_*\not\in \mathbb{A}_i) \mbox{ or }
				(u_i\in R_j) \mbox{ or }
				(T_*>T_j).\]
			$\mathcal{C}$ selects randomly a bit $b\in \{0,1\}$ and runs \singlefunction{Encrypt}$(\mathit{\mathrm{PI}},\mathit{\mathrm{PK}},M_*^b,S_*,T_*)$.\\ The obtained output $\mathrm{CT}_*$ is returned to $\mathcal{A}$.
			
			\item[Query, II]
			This phase is similar\footnote{with the additional restriction that any key which would allow $\mathcal{A}$ to simply use \singlefunction{Decrypt} cannot be requested. The querying process is again adaptive.} to the previous querying phase.			
			$\mathcal{A}$ can continue to query $\mathcal{C}$ a polynomial number of times with respect to~$\lambda$ with two possible requests.
			\begin{enumerate}
				\item
				(private-key query)
				$\mathcal{A}$ chooses a pair $(\mathbb{A}_i,u_i)$, where  $u_i$  was not previously requested, such that 
				\begin{itemize}
				\item $S_*\not\in \mathbb{A}_i$, or
				\item for each $k\leq j$ such that $T_*\leq T_k$, then $u_i\in R_k\,$.
				\end{itemize}
				$\mathcal{C}$ runs \singlefunction{GenKey}$(\mathit{\mathrm{PI}},\mathit{\mathrm{PK}},\mathit{\mathrm{MK}},\mathbb{A}_i,u_i)$ and gives the output $\mathit{\mathrm{SK}}_{\mathbb{A}_i,u_i}$ to~$\mathcal{A}$.
				
				\item
				(time-update query)
				$\mathcal{A}$ selects a pair $(T_j,R_j)$, where $R_j$ is an arbitrary set of users and $T_j$ is a time that was no previous requested for a previous time-update key (${\forall k < j},\:\allowbreak {T_j\neq T_{k}}$) such that
				\begin{itemize} 
				\item $T_*>T_j$, or
				\item for each $k\leq i$ such that $S_*\in \mathbb{A}_k$, then $u_k\in R_j\,$.
				\end{itemize}
				$\mathcal{C}$ runs \singlefunction{UpdateKey}$(\mathit{\mathrm{PI}},\mathit{\mathrm{PK}},\mathit{\mathrm{MK}},T_j,R_j)$ and gives the output  $\mathit{\mathrm{TK}}_{T_j,R_j}$  to $\mathcal{A}$.
			\end{enumerate}

			\item[Guess]
			$\mathcal{A}$ outputs a guess $b_*$ for the value of $b$.
		\end{description}
		If the output is correct, the adversary wins the game.		
		The \emph{advantage} of the adversary $\mathcal{A}$ is defined as the likelihood of making the right guess for $b$:
		\[
			\adv_{\mathcal{A}}^{\textnormal{RS-ABE}}(\lambda)=\left|\prob{b=b_*}-\frac{1}{2}\right|.
		\]
		Since the output of the functions shaping the scheme may depend on randomly chosen parameters, the advantage is described in term of the probability over all possible choices for their output considering their distribution.
	\end{definition}	
	\begin{definition}
 		A function $\eta(\lambda)$ is \emph{negligible} in $\lambda$ if for every $c>0$ and for every $k>0$ there exists $\lambda_0>0$ such that $|\eta(\lambda)| < \left|\frac{1}{c \lambda^{k}}\right|$ for all $\lambda > \lambda_0$.
 	\end{definition}
	\begin{definition}[Security]
		The RS-ABE scheme is said to be \emph{IND-CPA-secure}, or \emph{secure for a chosen-plaintext attack}, if the advantage of a PPT adversary is negligible with respect to the security parameter $\lambda$.
	\end{definition}

	Notice how the use of \singlefunction{Update$\mathrm{CT}$} does not allow the adversary to request a key that would allow the decryption of an updated ciphertext, because of the restrictions of the challenge phase and second-query phase. 
	
	Notice that the update key is not required to be secret.
	
	The aim of this paper is to propose a key-policy RS-ABE scheme which can be shown to be secure for a chosen-plaintext attack if the three assumptions in the next section hold.

	\section{Assumptions}
	\label{assumptions.sec}
	We work in the context of composite-order bilinear groups, adapted for the special case $N=p_1p_2p_3$, where each $p_i$ is a prime number. This is a concept first introduced by Boneh et al.~\cite{kilian2005bilineargroups}.
	\begin{definition}[Bilinear group]
		\label{bilineargroup.def}
		A \emph{bilinear group} (of composite order three) is identified by the tuple
		$
			((N,\G,\GT,e),p_1,p_2,p_3),
		$
		 where:
		
		\begin{enumerate}
			\item
			$p_1$, $p_2$ and $p_3$ are three distinct primes, and $N=p_1 p_2 p_3$;
			\item
			$\G$ and $\widetilde{\G}$ are two cyclic groups of order~$N$;
			\item
			$e:\,\G\times\G\rightarrow\widetilde{\G}$ is a map such that:
			\begin{itemize}
				\item
				\definline{bilinear} for every $h,k\in\G$ and $\alpha,\beta\in\Z_N$, 
				$e(h^\alpha,k^\beta)=e(h,k)^{\alpha\beta}; $
				\item
				\definline{non-degenerate} there exists $g\in\G$ such that $e(g,g)$ generates~$\widetilde{\G}$.
			\end{itemize}
		\end{enumerate}
		\end{definition}
		\begin{definition}
		\label{groupdes.def}
		A \emph{group descriptor} (of composite order three) is\\ a tuple $((N,\G,\GT,e),p_1,p_2,p_3,\bar{g}_1,\bar{g}_2,\bar{g}_3)$ such that
		\begin{itemize}
		\item the tuple $((N,\G,\GT,e),p_1,p_2,p_3)$ is a bilinear group;
		\item the group operations of~$\G$ and~$\GT$ and the map $e$ are efficiently computable;
		\item $\bar{g}_1$ is a generator of $\G_{p_1}$, $\bar{g}_2$ is a generator of $\G_{p_2}$ and $\bar{g}_3$ is a generator of $\G_{p_3}$, where $\G_{m}$ is the subgroup of~$\G$ with order~$m$.
		\end{itemize}
		\end{definition}
		\begin{remark}
		The three generators $\bar{g}_1$, $\bar{g}_2$, and $\bar{g}_3$ are only used implicitly to generate random element of the group $\G$ and its subgroups $\G_{p_1}$, $\G_{p_2}$, and $\G_{p_3}$.
		\end{remark}
		
		We observe that $g=\bar{g}_1 \bar{g}_2 \bar{g}_3$ generates $\G$, that $g^{p_2 p_3}$ is a generator of~$\G_{p_1}$ and that every element of~$\G_{p_1}$ has the form $g^{\alpha p_2 p_3}$ for some~$\alpha\in\N$.
		\begin{definition}
		\label{groupdesgen.def}
		A \emph{group descriptor generator} $\mathscr{G}$ is a polynomial algorithm which takes as input the security parameter~$\lambda$ and outputs a group descriptor \linebreak $\mathscr{G}(\lambda)=((N,\G,\GT,e),p_1,p_2,p_3,\bar{g}_1,\bar{g}_2,\bar{g}_3)$ such that 
	the group operations and the bilinear map  are computable in polynomial time with respect to~$\lambda$.
	\end{definition}
	
	
	The group descriptor $\mathscr{G}(\lambda)$ contains all the parameters used by the challenger for generating random elements of the group  $\G$ and its subgroups. The adversary knows only the tuple $(N,\G,\GT,e)$, which allows him to compute group operations and the bilinear map.
	
	The  following three assumptions fix the properties which must hold for a group descriptor generator $\mathscr{G}$, used in \singlefunction{Setup} (page~\pageref{RS-ABE_Setup-theory}).
	With the notation
	$
		r \leftarrow R
	$
	we specify that we are choosing an element $r$ among the elements of $R$ using a uniform random distribution. These assumptions were introduced by Lewko et al.~\cite{lewko2010assumptions}.
	Notice that in the following 
	the element $g_i$ belongs to $\G_{p_i}$, and 
	the element identified with the letters $X$, $Y$, and $Z$ belong respectively to $\G_{p_1}$, $\G_{p_2}$, and $\G_{p_3}$.

	\begin{assumption}[Subgroup decision problem]
		\label{a1}
		Let $\mathscr{G}$ be a group descriptor generator and for any $\lambda$ compute $\mathscr{G}(\lambda)$.
		We choose randomly the following elements:
			$g_1 \leftarrow{} \G_{p_1}$;
			$g_3 \leftarrow{} \G_{p_3}$;
		and we call $D=\left((N,\G,\GT,e),g_1,g_3\right)$. 
		We choose randomly two other elements:
			$W_1 \leftarrow{} \G_{p_1}$;
			$W_2 \leftarrow{} \G_{p_1p_2}$.
		The advantage of an algorithm $\mathscr{A}$ is defined as:
		\[
			\left|\prob{\mathscr{A}(D,W_1)=1}-\prob{\mathscr{A}(D,W_2)=1}\right|.
		\]
		This advantage is denoted with $\adv^{\textnormal{A1}}_{\mathscr{A}}$.
		We say that $\mathscr{G}$ satisfies Assumption \ref{a1} if for any PPT algorithm $\mathscr{A}$ we have that $\adv^{\textnormal{A1}}_{\mathscr{A}}$ is a negligible function with respect to $\lambda$.
	\end{assumption}
	\noindent The previous assumption could be informally rephrased as follows.
An adversary $\mathscr{A}$ returns $1$, i.e $\mathscr{A}(D,W)=1$, if and only if she thinks that $W$ is 
in the small subgroup $\G_{p_1}$, knowing that in any case $W$ is in the large subgroup $\G_{p_1p_2}$.
Her advantage is a measurement of the correctness of her guess.
In the case she is randomly guessing, she would return $1$ half of the times and so her advantage
would be zero. What we want with this assumption is to make any adversary at most negligibly better than a random guess.
	\begin{assumption}[General subgroup decision problem]
		\label{a2}
		Let $\mathscr{G}$ be a group descriptor generator and for any $\lambda$ compute $\mathscr{G}(\lambda)$. 
		We choose randomly the following elements:
			$g_1,X_1 \leftarrow{} \G_{p_1}$;
			$Y_1,Y_2 \leftarrow{} \G_{p_2}$;
			$g_3,Z_1 \leftarrow{} \G_{p_3}$;
		and we call \linebreak $D=\left((N,\G,\GT,e),g_1,g_3,X_1Y_1,Y_2Z_1\right)$.
		We choose randomly two other elements:
			$W_1 \leftarrow{} \G;$
			$W_2 \leftarrow{} \G_{p_1p_3}$.
		The advantage of an algorithm $\mathscr{A}$ is defined as:
		\[
			\left|\prob{\mathscr{A}(D,W_1)=1}-\prob{\mathscr{A}(D,W_2)=1}\right|.
		\]
		This advantage is denoted with $\adv^{\textnormal{A2}}_{\mathscr{A}}$.
		We say that $\mathscr{G}$ satisfies Assumption \ref{a2} if for any PPT algorithm $\mathscr{A}$ we have that $\adv^{\textnormal{A2}}_{\mathscr{A}}$ is a negligible function with respect to $\lambda$.
	\end{assumption}
	\noindent In this assumption, $\mathscr{A}$ is trying to guess if $W$ is  in the subgroup $\G_{p_1p_3}$, knowing that in any case $W$ is in  group $\G$.
	\begin{assumption}[Composite Diffie-Hellman problem]
		\label{a3}
		Let $\mathscr{G}$ be a group descriptor generator and for any $\lambda$ compute $\mathscr{G}(\lambda)$.
		We choose randomly the following elements:
			$\alpha,s \leftarrow{} \Z_{N}$;
			$g_1 \leftarrow{} \G_{p_1}$;
			$g_2,Y_1,Y_2 \leftarrow{} \G_{p_2}$;
			$g_3 \leftarrow{} \G_{p_3}$;
		and we call \linebreak 
		$D=\left((N,\G,\GT,e),g_1,g_2,g_3,g_1^\alpha Y_1,g_1^s Y_2\right)$. 
		We choose randomly another element
			$W_1 \leftarrow{} \GT$,
		and we fix $W_2=e(g_1,g_1)^{s\alpha}$.
		The advantage of an algorithm $\mathscr{A}$ is defined as:
		\[
			\left|\prob{\mathscr{A}(D,W_1)=1}-\prob{\mathscr{A}(D,W_2)=1}\right|.
		\]
		This advantage is denoted with $\adv^{\textnormal{A3}}_{\mathscr{A}}$.
		We say that $\mathscr{G}$ satisfies Assumption \ref{a3} if for any PPT algorithm $\mathscr{A}$ we have that $\adv^{\textnormal{A3}}_{\mathscr{A}}$ is a negligible function with respect to $\lambda$.
	\end{assumption}
	



	\section{Building blocks}\label{sec:block}
	
	We build our original key policy RS-ABE, starting from the ideas described by Lee et al.~\cite{lee2013RSABE} and adapting them to the key-policy case. To construct the scheme we outline three different schemes, which are later merged to form RS-ABE. These schemes are the \emph{complete subset} scheme (\emph{CS})~\cite{naor2001SC}, a \emph{self-updatable encryption} scheme (\emph{SUE})~\cite{lee2013RSABE}, and an \emph{attribute-based encryption} scheme (\emph{ABE})~\cite{lewko2010ABE}.
	
	\subsection{Notation}
	In the following construction we employ 
	perfect binary trees. 
	A perfect binary tree $\tree{}$ of depth $d$ is a binary tree where all leaf nodes has depth $d$, and all other nodes have exactly two children. Obviously,  exactly $2^d$ leaves hang from $\tree{}$. Each node of $\tree{}$ is denoted  $\nu_{i}$, where the index $i$ is assigned using breadth-first search on the tree. In other words, the root node has index $0$, the left child of $\nu_i$ has index $2i+1$ and the right child has index~$2i+2$.\\ 
	For each node $\nu_i$ we call $S_i$ the maximal subtree of $\tree{}$ having $\nu_i$ as root node. Clearly, $S_i$ is a perfect binary tree. To say that a node $\nu$ belongs to a subtree $S$ we write $\nu\in S$.
	
	In our case, each user in $\mathfrak{U}=\{u_1,\ldots,u_{N_{\textnormal{max}}}\}$  is associated to a leaf of the tree and vice versa. Sometimes, $u_j$  will denote the leaf corresponding to the user. The set $U_i\subset \mathfrak{U}$ contains all users associated with a leaf of the subtree $S_i$, in other words, all users whose leaf meets the node $\nu_i$ in the path from the leaf to the root node (see Figure \ref{fig1}).

	\begin{figure}[H]
	\begin{center}
	\forestset{
		nice empty nodes/.style={
		delay={where content={}{shape=coordinate,for parent={for children={anchor=north}}}{}}
	}}
	\begin{forest}
		grow down/.style={for tree={%
		  grow=south,
		  s sep=.5cm,
		  parent anchor=south,
		  child anchor=north
		  }
		}
		[ , draw, for tree={   l sep=0em, 
									l=3.5em,	   
									grow down,
									edge={dashed},
									if={n_children==0} 
										{circle, 
											draw, 
											minimum size=1.2em,
											l sep=1.6em,
											inner sep=0
										}
										{
										}
								}, 
							nice empty nodes ,
	   [ $\nu_i$ , draw,
	   tikz={\node [draw,rounded corners,line width=\boxPVWidth,boxPV,fit=(!11)(!22)] {};} 
	   [, for tree={consideredEdge,edge={solid}, line width=\consideredWidth} [  $u_1$ ] 
		   [  $u_2$ ] ]
	     [ , for tree={consideredEdge,edge={solid}, line width=\consideredWidth},  draw ,
	   [  $u_3$ ,
	   , ] 
		   [  $u_4$ ] ] ]
	   [ [ [  $u_5$ ] 
		   [  $u_6$ ] ]
	     [ [  $u_7$ ] 
		   [  $u_8$ ] ] ] ]
	\end{forest}
	\end{center}
		\caption{The solid line represents the subtree $S_i$ with root $\nu_i$ and the box represents the set $U_i$.}
		\label{fig1}
	\end{figure}
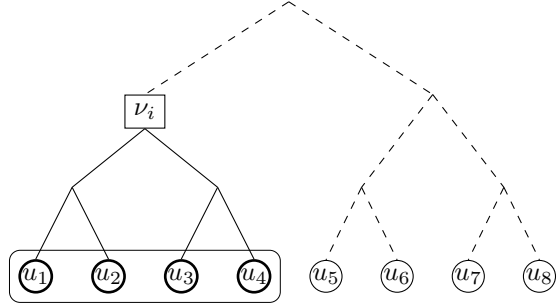
	
	 The \emph{Steiner tree} $S=S(\tree{},R)$, for a tree $\tree{}$ and a proper set of users $R\subset \mathfrak{U}$, is the smallest subtree of $\tree{}$ containing the root node and the leaves associated to all users in $R$. Observe that a Steiner tree is binary perfect only if $R=\emptyset$.
	
	\subsection{Complete subset scheme}
	\label{cs.sec}
	The complete-subset scheme is a particular implementation of the subset cover framework, introduced by Naor et al.~\cite{naor2001SC}. We use the same implementation described by Lee et al.~\cite{lee2013RSABE}.
	Notice that the subset-cover scheme as we describe it is not a full-fledged encryption mechanism, but rather a framework upon which we build encryption. Our main use of a CS scheme is to handle revoked users.\\
	The CS scheme is described by the following four functions.
	\begin{description}
	\algorithmdefsingleoutput{CS}{Setup}{N_{\textnormal{max}}}{\tree{}} we take the number of users $N_{\textnormal{max}}=2^d$ as an input and we build a perfect binary tree $\tree$ with depth $d$.
		Then we assign to each user in $\mathfrak{U}$ a unique leaf of the tree. The tree $\tree{}$ is the output of this function.

		\algorithmdefsingleoutput{CS}{Assign}{\tree{},u}{\PV_{u}} given the users' tree $\tree{}$ and a user $u\in \mathfrak{U}$, we output as \emph{private set} the subsets in $\mathfrak{U}$ of the form $U_i$ which contains $u$:
		\[
			\PV_u=\{U_i \mid u\in U_i\}\subset\mathcal{P}(\mathfrak{U}),
		\]
		where $\mathcal{P}(\mathfrak{U})$ is the power set of $\mathfrak{U}$.\\
		Observe that $|\PV_u|=d$.
	
		\algorithmdefsingleoutput{CS}{Cover}{\tree{},R}{\CV_{R}} we consider the Steiner tree $S=S(\tree{},R)$. We denote with $\nu_{k_1},\ldots,\nu_{k_m}$ all children of a node in $S$ which are not in $S$. The \emph{covering set} $\CV_R$ is then (see Figure \ref{figure.cover}):
		\[
			\CV_R=\{U_{k_1},\ldots,U_{k_m}\}\subset\mathcal{P}(\mathfrak{U}).
		\]
		Notice that $\CV_R$  is a partition of $\mathfrak{U}\setminus R$.
		
		\begin{figure}[H]
			\begin{center}
	\forestset{
		nice empty nodes/.style={
		delay={where content={}{shape=coordinate,for parent={for children={anchor=north}}}{}}
	}}
	\begin{forest}
		grow down/.style={for tree={%
		  grow=south,
		  s sep=.5cm,
		  parent anchor=south,
		  child anchor=north
		  }
		}
		[ , for tree={   l sep=0em, 
									l=3.5em,	   
									grow down,
									if={n_children==0} 
										{circle, 
											draw,
											minimum size=1.2em,
											l sep=1.6em,
											inner sep=0
										}
										{
										}
								}, 
							nice empty nodes
	   [ , edge={revokedEdge, line width=\revokedWidth} [ , edge={revokedEdge, line width=\revokedWidth} [  $u_1$ , edge={dashed},
	   tikz={\node [draw,rounded corners,line width=\boxCVWidth,boxCV,fit=()] {};} ] 
		   [  $u_2$, edge={revokedEdge, line width=\revokedWidth}, fill=revokedNode ] ] [ $\nu_{k_2}$ , draw,edge={dashed}, tikz={\node [draw,rounded corners,line width=\boxCVWidth,boxCV,fit=(!1)(!2)] {};}[  $u_3$ ,edge={dashed} ] 
		   [  $u_4$ ,edge={dashed} ] ] ]
	   [ , edge={revokedEdge, line width=\revokedWidth} [ , edge={revokedEdge, line width=\revokedWidth} [  $u_5$, edge={revokedEdge, line width=\revokedWidth}, fill=revokedNode ] 
		   [  $u_6$ ,edge={dashed},
	   tikz={\node [draw,rounded corners,line width=\boxCVWidth,boxCV,fit=()] {};}
	    ] ]
	     [ , edge={revokedEdge, line width=\revokedWidth} [  $u_7$, edge={revokedEdge, line width=\revokedWidth}, fill=revokedNode ] 
		   [  $u_8$, edge={revokedEdge, line width=\revokedWidth}, fill=revokedNode ] ] ] ]
	\end{forest}
	\end{center}
	\caption{$R=\{u_2,u_5,u_7,u_8\}$. The solid subtree is the Steiner tree $S(\tree{},R)$, the gray leaves are the revoked users, $\nu_{k_1}=u_1$, $\nu_{k_2}\not\in\mathfrak{U}$, $\nu_{k_3}=u_6$, $U_{k_1}=\{u_1\}$, $U_{k_2}=\{u_3,u_4\}$, $U_{k_3}=\{u_6\}$.}
		\label{figure.cover}
		\end{figure}
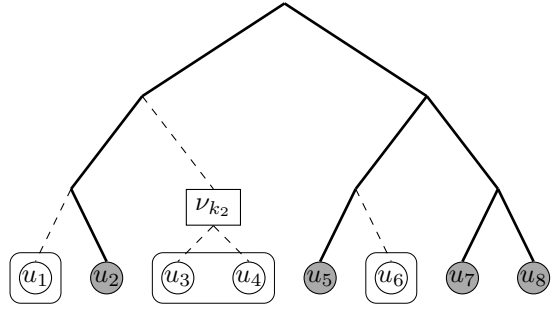
		
		\algorithmdefsingleoutput{CS}{Match}{\CV_R,\PV_u}{U} we find a match between a covering set $\CV_R$ and a private set $\PV_u$ simply by considering the (unique) element in the intersection of the two sets, $U\in \CV_R \cap \PV_u$, when it exists (see Figure \ref{fig3}). The output is $U$
		if $U$ exists, empty otherwise.
		A match exists if and only if the user is not revoked. It must be unique, since $\CV_R$ is a partition.
	
	\begin{figure}[H]
	\begin{center}
	\forestset{
		nice empty nodes/.style={
		delay={where content={}{shape=coordinate,for parent={for children={anchor=north}}}{}}
	}}
	\begin{forest}
		grow down/.style={for tree={%
		  grow=south,
		  s sep=.5cm,
		  parent anchor=south,
		  child anchor=north
		  }
		}
		[ , for tree={   l sep=0em, 
									l=3.5em,	   
									grow down,
									if={n_children==0} 
										{circle, 
											draw, 
											minimum size=1.2em,
											l sep=1.6em,
											inner sep=0
										}
										{
										}
								}, 
							nice empty nodes
	   [ , edge={revokedEdge, line width=\revokedWidth} [  , edge={revokedEdge, line width=\revokedWidth} [  $u_1$, edge={dashed} ] 
		   [  $u_2$ , edge={revokedEdge, line width=\revokedWidth} , fill=revokedNode ] ]
	     [ , edge={consideredEdge, dotted, line width=\consideredWidth},
	   tikz={\node [draw,rounded corners,line width=\boxSetWidth,boxSet,fit=(!1)(!2)] {};} [  $u_3$ , edge={consideredEdge, dotted, line width=\consideredWidth}, fill=consideredNode] 
		   [  $u_4$ , edge={dashed} ] ] ]
	   [ , edge={revokedEdge, line width=\revokedWidth} [ , edge={revokedEdge, line width=\revokedWidth} [  $u_5$, edge={revokedEdge, line width=\revokedWidth}, fill=revokedNode ] 
		   [  $u_6$ , edge={dashed}
	    ] ]
	     [ , edge={revokedEdge, line width=\revokedWidth} [  $u_7$, edge={revokedEdge, line width=\revokedWidth}, fill=revokedNode ] 
		   [  $u_8$, edge={revokedEdge, line width=\revokedWidth}, fill=revokedNode ] ] ] ]
	\end{forest}
	\end{center}
	\caption{Considering Figure~\ref{figure.cover}, the box is the output of \singlefunction{CS.Match} for $\PV{}_{u_3}=\big\{ \{u_3\},\{u_3,u_4\},\{u_1,u_2,u_3,u_4\},\mathfrak{U}  \big\}$. }
	\label{fig3}
	\end{figure}
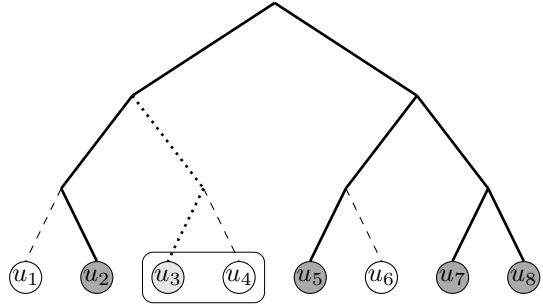
	
	\end{description}
	\begin{remark}
	It is worth considering the size of the private set and of the covering set. Recall that $N_{\textnormal{max}}=2^{d}$ and let $|R|=r$ then~\cite{naor2001SC}
	\[
	|\PV_u|=\log{}_2(N_{\textnormal{max}})=d;
	\]\[
	  |\CV_R|\leq r \log{}_2(N_{\textnormal{max}}/r)=r(d-\log{}_2(r)).
	\]
	
			\end{remark}
		
	\subsection{Self-updatable encryption scheme}
The self-updatable encryption scheme (SUE) manages the evolution of the system, introducing a discrete time structure used to equip ciphertexts and keys with specific times. The basic idea behind this scheme is to associate (with depth-first search) each time with a node of a perfect binary tree and then to employ the ciphertext-delegatable encryption scheme (CDE)~\cite{lee2013RSABE} to validate the time of the given decryption key, exploiting the tree structure. This allows to identify the nodes associated with a time bigger than a certain threshold, by using a small number of keys.\\
Here we present  a high-level description of  the SUE scheme.
The SUE scheme is characterized by the following five PPT algorithms, the detailed description is provided by Lee et al.~\cite{lee2013RSABE}.
\begin{description}

	\algorithmdef{SUE}{Setup}{\mathscr{S},T_\textnormal{max}}{\mathrm{MK},\mathrm{PI},\mathrm{PK}} given a string\footnote{we might  as well give a group descriptor.} $\mathscr{S}=((N, \G, \GT, e), g_1,p_1,p_2,p_3)$ and the maximum admissible time $T_\textnormal{max}\in\mathbb{N}$, it outputs the master  key $\mathrm{MK}$, which is used by the authority $\mathcal{C}$ to create private keys,   the public information $\mathrm{PI}$ and the (general) public key $\mathrm{PK}$. In particular, we choose randomly the elements $\beta \in \Z_N$ and $Z\in\G_{p_3}$. The outputs of the function are
	$\mathrm{MK}=( \beta, Z )$, $\mathrm{PI}= (N,\G,\GT,e)$  
	and the (general) public key $\mathrm{PK}$, that is, a string containing $g=g_1$, $e(g,g)^\beta$ and some randomly chosen elements in $\G_{p_1}$.

		\algorithmdefsingleoutput{SUE}{GenKey}{\mathrm{PI},\mathrm{PK},\mathrm{MK},T^{\prime}}{\mathrm{SK}_{T^{\prime}}} given the public information $\mathrm{PI}$, the public key $\mathrm{PK}$, the master key $\mathrm{MK}$ and a time $T^{\prime}$, it outputs a private key $\mathrm{SK}_{T^{\prime}}$ associated  $T^{\prime}$.

		\algorithmdef{SUE}{Encrypt}{\mathrm{PI},\mathrm{PK},T,M}{C,\mathrm{CH}_{T}} given the public information $\mathrm{PI}$, the public key $\mathrm{PK}$, a time $T$ and a plaintext $M$, it outputs the ciphertext $C$ and the  ciphertext header $\mathrm{CH}_{T}$ associated with the time $T$. In particular, we choose randomly an element $s\in \Z_N$ and a session key $\mathrm{EK}=e(g,g)^{s\beta}$ is computed. With $\mathrm{EK}$, we compute $C=M\cdot\mathrm{EK}$ and we do not output $\mathrm{EK}$.

		\algorithmdef{SUE}{Decrypt}{\mathrm{PI},\mathrm{SK}_{T^{\prime}},\mathrm{CH}_{T},C}{\mathrm{EK},M} given
		the public information $\mathrm{PI}$, a ciphertext $(C,\mathrm{CH}_{T})$ for the time $T$ and a secret key $\mathrm{SK}_{T^{\prime}}$ for the time $T^{\prime}$, first it computes $\mathrm{EK}$ from the ciphertext header $\mathrm{CH}_{T}$ and then  $M=C\cdot\mathrm{EK}^{-1}$. Note that the output  might be empty in case the given key is not allowed to decrypt the ciphertext.
		
		\algorithmdefsingleoutput{SUE}{Update$\mathrm{CT}$}{\mathrm{PI},\mathrm{PK},\mathrm{CH}_{T}}{\mathrm{CH}_{T+1}}
		given the public information $\mathrm{PI}$, the public key $\mathrm{PK}$ and a ciphertext $\mathrm{CH}_{T}$ associated with a time $T$, it outputs a ciphertext $\mathrm{CH}_{T+1}$ associated with the time $T+1$, which decrypts to the same plaintext of $\mathrm{CH}_{T}$

		\algorithmdefsingleoutput{SUE}{Randomize$\mathrm{CT}$}{\mathrm{PI},\mathrm{PK},\mathrm{CH}_{T}}{\overline{\mathrm{CH}}_{T}} given the public information $\mathrm{PI}$, the public key $\mathrm{PK}$ and a ciphertext $\mathrm{CH}_{T}$ for the time $T$, it outputs a  ciphertext $\overline{\mathrm{CH}}_{T}$ for the same time $T$.
		This function takes a ciphertext and transforms it to another one, associated with the same time, but which decrypts to the same plaintext.\footnote{The function is named ``Randomize" because it changes the internal randomness used to create the ciphertext.} 
		
	\end{description}
	
	\noindent The \emph{correctness} of this scheme is characterized by a correct decryption:
	\[
		\singlefunction{SUE.Decrypt}(\mathrm{PI},\mathrm{SK}_{T^{\prime}},\mathrm{CH}_T,C)=(\mathrm{EK},M)
	\]
	when $T$ is a time less or equal than $T^{\prime}$, or otherwise empty.
	This must be true for every output of \singlefunction{SUE.Setup}, any valid times $T$, $T^{\prime}$, any ciphertext $(C,\mathrm{CH}_T)$ generated by \singlefunction{SUE.Encrypt} and any private key $\mathrm{SK}_{T^\prime}$ generated by \singlefunction{SUE.GenKey}.
	We require also that for any ciphertext $(C,\mathrm{CH}_T)$, the decryption of $(C,\mathrm{CH}_T)$ offers the same plaintext as the 
	the decryption of $(C,$\singlefunction{SUE.Randomize$\mathrm{CT}$}$)$.
	The distribution of the outputs of \singlefunction{SUE.Randomize$\mathrm{CT}$} must be equal to the distribution of the outputs of \singlefunction{SUE.Encrpyt}.
	
	\subsection{Attribute-based encryption scheme}
	The attribute-based encryption scheme (ABE) regulates the access to ciphertexts with respect to the presence or absence of specific attributes.
	ABE is an encryption scheme delineated for the first time by Sahai and Waters~\cite{sahai2005}. The first fully secure ABE scheme is described by Lewko et al.~\cite{lewko2010ABE}, and this is the scheme we are going to employ. Various other results for ABE are described in~\cite{chase2007, garg2013, waters2011}, and a non-monotonic ABE scheme in~\cite{ostrovsky2007}.
	\begin{definition}[Access structure]
	\label{def.accessstructure}
	An \emph{attribute} is an element of a given finite set $\mathfrak{A}=\left\{a_1,\ldots,a_n\right\}$, the set of all attributes. An \emph{access structure} is a non-empty subset of $ \powerset{\mathfrak{A}}$. 
	An access structure $\mathbb{A}$ is \emph{monotonic}  if for any $S\in\mathbb{A}$  we have that also
	all supersets of $S$ belong to $\mathbb{A}$. 
	\end{definition}
	\noindent From now on we consider implicitly only monotonic access structure.
	
	We can find two different forms of ABE schemes, depending on the position of attributes and access structure between the key and the ciphertext. The differentiation was first stated in~\cite{goyal2006}.
	\begin{itemize}
		\item
		A \emph{ciphertext-policy} (\emph{CP}) ABE scheme provides  each ciphertext with an access structure $\mathbb{A}$ (from which the name ``ciphertext-policy''), and  each secret key with a set of attributes $S$.
		\item
		A \emph{key-policy} (\emph{KP}) ABE scheme does exactly the converse: it provides  each secret key with an access structure $\mathbb{A}$, and  each ciphertext with a set of attributes $S$.
	\end{itemize}
	A general method to transform a KP-ABE in a CP-ABE scheme is discussed in~\cite{goyal2008}.	
	We are going to describe the KP version of the scheme of Lewko \cite{lewko2010ABE}, since we are going to construct our RS-ABE scheme from it. A key-policy RS-ABE scheme was independently described by \citeauthor{lee2014}~\cite{lee2014}; however the author employs a different construction and different assumptions to prove the security of the scheme.
\subsubsection{Linear Secret Sharing Scheme.}	
Before describing the construction of the ABE algorithm, we have to recall the construction of a \emph{linear secret sharing scheme (LSSS)} that is used in order to manage the attribute structure. For more details about LSSS see~\cite{beimel1996SSS}.
		\label{LSSS_def}

	The classical construction, that we are going to present, distributes the shares to each \texttt{LSSS user}. However in an ABE scheme the role of the LSSS users is played by the attributes in $\mathfrak{A}$. So we will use the set $\mathfrak{A}$ for the set of LSSS users, with $a_h\in\mathfrak{A}$ denoting any LSSS user.

		Let $\G=\Z_{p_1p_2p_3}$ be the set of the possible secrets. A \emph{linear secret sharing scheme}  (\emph{LSSS}) for a secret $s\in \G$ is a scheme constructed with the following procedure. 
		\begin{enumerate}
			\item
			Let $m\in\mathbb{N}$ be a public parameter. The authority knows the secret $s$ and chooses a random element $\vec{r}=\left(r_2,\ldots,r_{m}\right)\in \G^{m-1}$ uniformly over $\G^{m-1}$. 
			\item
			 The authority assign to each LSSS user $a_i$   a vector in $\G^{d_i}$, called the \emph{vector of shares}. The $\{d_i\}_{1\leq i \leq 2^d}\subset\mathbb{N}$ are  public parameters.  Each vector entry, a \emph{share}, is a linear combination of $s$ and of $r_i$'s.
	\noindent We denote the $j\ith$ element of the shares vector of $a_h$ as $\pi_{h,j}(s,\vec{r})$, and we call $l=\sum_{a_h\in \mathfrak{A}}d_h$ the \emph{size} of the scheme.
		\end{enumerate}
	We can then identify a generic LSSS with a pair $(B,\rho)$ and a vector $\vec{v}$, where:
	
	\begin{itemize}
		\item
		$B$ is an $l\times m$ matrix used to generate all the shares.
		\item
		$\rho: \{1,\ldots,l \}\rightarrow \mathfrak{A}$ is a function which assigns to every row of the matrix its corresponding LSSS user.
		\item
		$\vec{v}=\left(s,r_2,\ldots,r_{m}\right)$ is a vector in $\G^{m}$ which has the secret as the first component, and the remaining ones are random elements of $\G$.
	\end{itemize}
The authority creates the shares multiplying the matrix $B$ by the vector $\vec{v}$. Then, each LSSS user $a$ receives all  corresponding components of the vector $B\cdot \vec{v}$, i.e., $a$ obtains from the authority the elements with position in the set $\{h\mid \rho(h)=a\}$ as shares. The sets of LSSS users who can obtain the secret are the \emph{authorized sets}. The authorized sets form an access structure $\mathbb{A}$ which is monotonic, since adding LSSS users does not diminish the number of shares available for retrieving the secret.

Any authorized set $S\in \mathbb{A}$ can reconstruct the secret via linear combination of their shares. This means that there exist some constants 
			\[ \{ \omega_{h,j} \in \G \mid a_h\in S, 1\leq j\leq d_h \} \]
			such that the secret $s$ can be obtained as
			\[ s=\sum_{a_h \in S} \sum_{j=1}^{d_h} \omega_{h,j} \cdot \pi_{h,j}(s,\vec{r}) . \]
			More details are provided by Beimel~\cite{beimel1996SSS}.

\subsubsection{Construction of the ABE scheme.}\label{ABE.construction}
First we have to suppose that the function $\rho$ is injective, i.e., to each attribute is assigned only a single share. This hypothesis is used in the proof of security of  RS-ABE\@ scheme (see Section~\ref{security_section}). This restriction can be lifted by enlarging the dimensions of the keys, as we are going to see in Section~\ref{efficiency.sec}. For describing the construction we use the same notation as in the previous sections.
\begin{description}
		\algorithmdef{ABE}{Setup}{\mathscr{S},\mathfrak{A}}{\mathrm{MK},\mathrm{PI},\mathrm{PK}} where  $\mathscr{S}=((N, \G, \GT, e), g_1,p_1,p_2,p_3)$ and  $\mathfrak{A}$ is a set of attributes.\\ 
We choose randomly the elements:
	\begin{itemize}
		\item $\gamma \in \Z_N$;
		\item $T_a\in \G_{p_1}\text{, for all }a\in \mathfrak{A}$;
		\item $Z\in\G_{p_3}$.
	\end{itemize}
	The outputs of the function are the master secret key $\mathrm{MK}$, the public information $\mathrm{PI}$,
	\[
		\mathrm{MK}=\left( \gamma, Z \right) \qquad \mathrm{PI}= \left(N,\G,\GT,e\right)
	\]
	and the (general) public key,
	\[
		\mathrm{PK}=\left( g=g_1, \left\{T_a\right\}_{a\in \mathfrak{A}}, \Lambda=e(g,g)^\gamma \right).
	\]
	
	\algorithmdefsingleoutput{ABE}{GenKey}{\mathrm{PI},\mathrm{PK},\mathrm{MK},\mathbb{A}}{\mathrm{SK}_{\mathbb{A}}} where $\mathbb{A}=(B,\rho)$ is a generic LSSS access structure  as in Section~\ref{LSSS_def}, of size $l$ for a set of attributes $\mathfrak{A}$, such that $B\in \matrixset{l }{m }{\Z_N}$  is a matrix with $l$ rows and $m$ columns.\\ 
	We choose randomly the elements:
\begin{itemize}
		\item $r_2,\ldots,r_m  \in \Z_N$;
		\item $s_1,\ldots,s_l  \in \Z_N$;
		\item $Z_{1,i},Z_{2,i}\in \G_{p_3}\text{, for all }1\leq i\leq l $.
	\end{itemize}
	We also define the vector $\vec{v}=\left( \gamma, r_2, \ldots, r_m  \right)$.\\
	The output is the secret key\footnote{Any secret key is actually a list of $l$ key pairs. Both keys in a pair lie in $\G_{p_1p_3}$.} associated with the access structure $\mathbb{A}$:
	\[
		\mathrm{SK}_{\mathbb{A}}=\left(\left\{K_{1,i}=g^{B_i\cdot \vec{v}} (T_{\rho(i)})^{s_i} Z_{1,i}, K_{2,i}=g^{s_i} Z_{2,i}\right\}_{1\leq i\leq l } \right),
	\]
	where $B_i$ indicates the $i\ith$ row of the matrix $B$.
	\algorithmdef{ABE}{Encrypt}{\mathrm{PI},\mathrm{PK},S,M}{C,\mathrm{CH}_{S}} where $S$ is a set of attributes and   $M$ is a plaintext. First, randomly chosen $s\in\mathbb{Z}_N$, it computes
	\[
		\mathrm{EK}=\Lambda^s\in \GT;
	\]
	then the outputs are $C=M\cdot\mathrm{EK}$ and the ciphertext header:
	\[
		\CH_{S}=\left(C_0=g^s,\left\{C_{1,a}=T_a^s\right\}_{a\in S} \right).
	\]

	\algorithmdef{ABE}{Decrypt}{\mathrm{PI},\mathrm{SK}_{\mathbb{A}},\CH_S,C}{\EK,M} the output is the same session key $\EK$ we have obtained when we ran \algorithm{ABE}{Encrypt} and the corresponding plaintext $M$. Since $S$ is authorized  for $\mathbb{A}$, we know that there exist some constants $\omega_j \in \Z_N$ such that 
	\[
		\sum_{\rho(i)\in S} \omega_i \cdot B_i = \left(1,0,\ldots,0\right).
	\]
	If we know the $\omega_i$'s. we can compute the session key:
	\begin{equation}\label{eq:EKABE}
		\EK=\prod_{\rho(i)\in S} \left( \frac{e(C_0,K_{1,i})}{e(C_{1,\rho(i)},K_{2,i})} \right)^{\omega_{i}};
	\end{equation}
	and then the plaintext $M=C\cdot \EK^{-1}$.

	\algorithmdef{ABE}{Randomize$\mathrm{CT}$}{\mathrm{PI},\mathrm{PK},\EK,\CH_{S},\overline{s}}{\EK^{\prime},\overline{\CH}_{S}} $\overline{s}\in \Z_N$.
	If we call the components of the given header as $\CH_{S}=\left( C_0, \left\{ C_{1,a}\right\}_{a\in S} \right)$ the outputs are the rerandomized session key:
	\[
		\overline{\EK}=\EK \cdot \Lambda^{\overline{s}};
	\]
	and the rerandomized ciphertext header:
	\[
	\overline{\CH}_{S}=\left(C_0\cdot g^{\overline{s}},\left\{C_{1,a}\cdot T_a^{\overline{s}}\right\}_{a\in S} \right)
	\]

	\end{description}
	
	The correctness of the ABE algorithm follows from the properties of the bilinear map $e$.
	
	\begin{align*}
		\prod_{\rho(i)\in S} \left( \frac{e(C_0,K_{1,i})}{e(C_{1,\rho(i)},K_{2,i})} \right)^{\omega_{i}} &= \prod_{\rho(i)\in S} \left( \frac{e(g^s,g^{B_i\cdot \vec{v}} T_{\rho(i)}^{s_i} Z_{1,i})}{e(T_{\rho(i)}^s,g^{s_i} Z_{2,i})} \right)^{\omega_{i}} \\
		&= \prod_{\rho(i)\in S} \left( \frac{e(g^s,g^{B_i\cdot \vec{v}})\cdot e(g^s,T_{\rho(i)}^{s_i})}{e(g^{s_i},T_{\rho(i)}^{s})} \right)^{\omega_{i}} \\
		&= \prod_{\rho(i)\in S} e(g^s,g^{B_i\cdot \vec{v}})^{\omega_{i}} \\
		&= \left( e(g,g) ^{\sum_{\rho(i)\in S}\omega_{i}B_i\cdot \vec{v}}\right)^s \\
		&= e(g,g)^{\gamma s} = \Lambda^{s}.
	\end{align*}
	
	\pagebreak
	\section{The scheme}\label{sec:scheme}
	
	In this section we describe our construction of a key-policy RS-ABE scheme and we show its correctness. We strongly recommend the reader to read the following scheme description while rereading Section \ref{sec:structure} since we use the notation we introduced before, providing the details required to build the scheme.
	
	\subsection{Construction}
	
	\begin{description}
	\algorithmdef{RS-ABE}{Setup}{\lambda,\mathfrak{A},T_{\textnormal{max}},N_{\textnormal{max}}}{\mathit{\mathrm{MK}},\mathit{\mathrm{PI}},\mathit{\mathrm{PK}}} 
	
	We take the following steps.
	\begin{enumerate}
		\item
		We use the security parameter to generate a group descriptor $\mathscr{G}(\lambda)$ (De\-finition~\ref{groupdes.def}) using a group descriptor generator $\mathscr{G}$ (Definition~\ref{groupdesgen.def}). We fix a generator $g=g_1$ of $\G_{p_1}$, and we define $\mathscr{S}=\left((N, \G, \GT, e), g_1,p_1,\allowbreak{}p_2,\allowbreak{}p_3\right)$.
		\item
		We use \algorithm{CS}{Setup}$(N_{\textnormal{max}})$ to obtain a perfect binary tree $\tree$ with $N_{\textnormal{max}}=2^d$ leaves.
		\randomchoose{p}$\gamma_{i}\in \Z_N$, for each node $\nu_i\in \tree$.
		Every node $\nu_i$ of the tree $\tree$ is associated with the element $\gamma_{i}$.
		\item
		We choose randomly the elements:
		\begin{itemize}
		\item $T_a\in \G_{p_1}\text{, for all }a\in \mathfrak{A}$;
		\item $Z\in\G_{p_3}$;
		\item $\alpha\in \Z_N$.
	\end{itemize}
		
		\item
		For each node $\nu_i\in\tree$ 
	we define 
	\[
		\mathrm{MK}_{\textnormal{ABE},i}=\left( \gamma_i, Z \right). 
	\]
	
	Note that we are considering almost the same construction of $\mathrm{MK}$ as \algorithm{ABE}{Setup}:  for each node $\nu_i\in\tree$ we have a common $Z$ but a different $\gamma_i$. We will denote $\mathrm{MK}_{\textnormal{ABE}}=\{\mathrm{MK}_{\textnormal{ABE},i}\}_{\nu_i\in\tree}$.
	\item
	We define
	\[
		\mathrm{PK}_{\textnormal{ABE}}=\left( g=g_1, \left\{T_a\right\}_{a\in \mathfrak{A}}\right).
	\]
Note that we are considering almost the same construction of $\mathrm{PK}$ as \algorithm{ABE}{Setup} but here   $\mathrm{PK}_{\textnormal{ABE}}$ does not contain $e(g,g)^{\gamma_i}$ for any $\nu_i\in\tree$.
\item
For each node $\nu_i\in\tree$ we define
	\[
		\mathrm{MK}_{\textnormal{SUE},i}=\left( \alpha-\gamma_i, Z \right).
	\]
	
	Note that we are considering almost the same construction of $\mathrm{MK}$ as \algorithm{SUE}{Setup}: for each node $\nu_i\in\tree$ we have a common $Z$ but a different $\gamma_i$ (the same considered before). We will denote $\mathrm{MK}_{\textnormal{SUE}}=\{\mathrm{MK}_{\textnormal{SUE},i}\}_{\nu_i\in\tree}$.
	\item
	We define $\mathrm{PK}_{\textnormal{SUE}}$ containing $g=g_1$ and some randomly chosen elements in $\G_{p_1}$.\\
	Note that we are considering almost the same construction of $\mathrm{PK}$ as \algorithm{SUE}{Setup} but here   $\mathrm{PK}_{\textnormal{SUE}}$ does not contain $e(g,g)^{\alpha-\gamma_i}$ for any $\nu_i\in\tree$.

		
	\end{enumerate}
	The final output is:
	\[
		\mathit{\mathrm{MK}}=\left( \alpha, \tree,\mathrm{MK}_{\textnormal{ABE}},\mathrm{MK}_{\textnormal{SUE}}\right),
	\]
	\[
		\mathit{\mathrm{PI}}=\left( N, \G, \GT, e\right),
	\]
	\[
		\mathit{\mathrm{PK}}=\left( \Omega=e(g,g)^{\alpha},  \mathit{\mathrm{PK}}_{\textnormal{ABE}}, \mathit{\mathrm{PK}}_{\textnormal{SUE}}\right).
	\]
	
	\algorithmdefsingleoutput{RS-ABE}{GenKey}{\mathit{\mathrm{PI}},\mathit{\mathrm{PK}},\mathit{\mathrm{MK}},\mathbb{A},u}{\mathit{\mathrm{SK}}_{\mathbb{A},u}} this function generates a secret key associated with the user index $u$ and an access structure $\mathbb{A}$, requiring as input also the public information, the (general) public key and the master secret key.
	\begin{enumerate}
		\item
		We run \algorithm{CS}{Assign}$(\tree,u)$ and obtain a private set $\PV_u=\{U_{k_1},\ldots,U_{k_d}\}$.
		\item
 		 We consider the ABE  public key $\mathit{\mathrm{PK}}_{\textnormal{ABE}}$ and for every $i$ from $1$ to $d$ the ABE master private key $\mathit{\mathrm{MK}}_{\textnormal{ABE},k_i}=(\gamma_{k_i},Z)$. We  obtain $\mathit{\mathrm{SK}}_{\textnormal{ABE},k_i}$ by running \algorithm{ABE}{GenKey}$(\mathit{\mathrm{PI}},\mathit{\mathrm{PK}}_{\textnormal{ABE}},$ $\mathit{\mathrm{MK}}_{\textnormal{ABE},k_i},\mathbb{A})$.
	\end{enumerate}
	Then we output the secret key:
	\[
		\mathit{\mathrm{SK}}_{\mathbb{A},u}=\left( \mathit{\mathrm{PV}}_u, \mathit{\mathrm{SK}}_{\textnormal{ABE},k_1}, \ldots, \mathit{\mathrm{SK}}_{\textnormal{ABE},k_d} \right).
	\]
	
	\algorithmdefsingleoutput{RS-ABE}{UpdateKey}{\mathit{\mathrm{PI}}, \mathit{\mathrm{PK}},\mathit{\mathrm{MK}},T,R}{\mathit{\mathrm{TK}}_{T,R}} this function generates the key $\mathit{\mathrm{TK}}_{T,R}$ which associates with a certain time $T$ the corresponding set of revoked users $R$. It requires  the public information, the public key and the master private key of the scheme.
	\begin{enumerate}
		\item
		We run \algorithm{CS}{Cover}$(\tree,R)$ and obtain a covering $\mathrm{CV}_R=\{U_{k_0},\ldots,U_{k_m}\}$, for some $m$, and every set $U_{k_i}$ is the set of the leaves of the subtree $S_{\nu_{k_i}}$.
		\item
		We consider the SUE public key $\mathit{\mathrm{PK}}_{\textnormal{SUE}}$ and for every $i$ from $0$ to $m$  the SUE master private key $\mathit{\mathrm{MK}}_{\textnormal{SUE},k_i}=(\alpha - \gamma_{k_i},Z)$. We  obtain $\mathit{\mathrm{SK}}_{\textnormal{SUE},k_i}$ by running \algorithm{SUE}{GenKey}$(\mathit{\mathrm{PI}},\mathit{\mathrm{PK}}_{\textnormal{SUE}}, \mathit{\mathrm{MK}}_{\textnormal{SUE},k_i},T)$.
	\end{enumerate}
	Then we output the time key:
	\[
		\mathit{\mathrm{TK}}_{T,R}=\left( \mathit{\mathrm{CV}}_R, \mathit{\mathrm{SK}}_{\textnormal{SUE},k_0}, \ldots, \mathit{\mathrm{SK}}_{\textnormal{SUE},k_m} \right).
	\]
	
	\algorithmdefsingleoutput{RS-ABE}{Encrypt}{\mathit{\mathrm{PI}},\mathit{\mathrm{PK}},M,S,T}{\mathrm{CT}_{S,T}} we require as input the public information, the public key, a  plaintext $M\in \GT$, a set of attributes $S$ and a time $T$. The output is a ciphertext $\mathrm{CT}_{S,T}$.
	\randomchoose{s}$s\in \Z_N$.
	Then we run \algorithm{ABE}{Encrypt}$(\mathit{\mathrm{PI}},\mathit{\mathrm{PK}}_{\textnormal{ABE}},S,M)$ and \algorithm{SUE}{Encrypt}$(\mathit{\mathrm{PI}},\mathit{\mathrm{PK}}_{\textnormal{SUE}},T,M)$ in order to obtain $\mathit{\mathrm{CH}}_{\textnormal{ABE}}$ and $\mathit{\mathrm{CH}}_{\textnormal{SUE}}$. 
	The output ciphertext is:
	\[
		\mathrm{CT}_{S,T}=\left( \mathit{\mathrm{CH}}_{\textnormal{ABE}},\mathit{\mathrm{CH}}_{\textnormal{SUE}},C=\Omega^s \cdot M \right).
	\]
	
	\algorithmdefsingleoutput{RS-ABE}{Decrypt}{\mathit{\mathrm{PI}},\mathrm{CT}_{S,T},\mathit{\mathrm{SK}}_{\mathbb{A},u},\mathit{\mathrm{TK}}_{T^{\prime}\!,R}}{M} we require as input the public information, a ciphertext $\mathrm{CT}_{S,T}=\left( \mathit{\mathrm{CH}}_{\textnormal{ABE}}, \mathit{\mathrm{CH}}_{\textnormal{SUE}}, C \right)$ for a time $T$ and a set of attributes $S$, a private key $\mathit{\mathrm{SK}}_{\mathbb{A},u}=\left( \mathit{\mathrm{PV}}_u, \mathit{\mathrm{SK}}_{\textnormal{ABE},k_1}, \ldots, \mathit{\mathrm{SK}}_{\textnormal{ABE},k_d} \right)$ for a user $u$ and an LSSS access structure $\mathbb{A}$ which contains $S$, a time key $\mathit{\mathrm{TK}}_{T^{\prime}\!,R}=\left( \mathit{\mathrm{CV}}_R, \mathit{\mathrm{SK}}_{\textnormal{SUE},k_0}, \ldots, \mathit{\mathrm{SK}}_{\textnormal{SUE},k_m} \right)$ for a time $T^{\prime}\geq T$ and a set of revoked users such that $u\not\in R$. The output is the  plaintext $M$.
	
	\begin{enumerate}
		\item
		We run \algorithm{CS}{Match}$(\mathit{\mathrm{CV}}_R,\mathit{\mathrm{PV}}_u)$ and obtain a set of users of the form $U_k$.
		\item
		We run \algorithm{ABE}{Decrypt}$(\mathit{\mathrm{PI}},\mathit{\mathrm{SK}}_{\textnormal{ABE},k},\mathit{\mathrm{CH}}_{\textnormal{ABE}})$ to obtain $\mathit{\mathrm{EK}}_{\textnormal{ABE},k}=e(g,g)^{\gamma_k s}$ and  \algorithm{SUE}{Decrypt}$(\mathit{\mathrm{PI}},\mathit{\mathrm{SK}}_{\textnormal{SUE},k},\mathit{\mathrm{CH}}_{\textnormal{SUE}})$ to obtain $\mathit{\mathrm{EK}}_{\textnormal{SUE},k}=e(g,g)^{(\alpha  - \gamma_k) s}$, which are the ciphertext headers used for the decryption in ABE and SUE. By construction,  the ABE ciphertext header and the SUE ciphertext header do not depend on $\gamma_k$, so it is not necessary to include the index $k$ in $\mathit{\mathrm{CH}}_{\textnormal{ABE}}$ and $\mathit{\mathrm{CH}}_{\textnormal{ABE}}$.
	\end{enumerate}
	Then we output the plaintext:
	\[
		C\cdot (\mathit{\mathrm{EK}}_{\textnormal{ABE},k}\cdot \mathit{\mathrm{EK}}_{\textnormal{SUE},k})^{-1}.
	\]
	
	\algorithmdefsingleoutput{RS-ABE}{Update$\mathrm{CT}$}{\mathit{\mathrm{PI}},\mathit{\mathrm{PK}},\mathrm{CT}_{S,T}}{\mathrm{CT}_{S,T+1}} the inputs are the public information, the public key and a ciphertext.
	The output is a ciphertext $\mathrm{CT}_{S,T+1}$ which encrypts the same plaintext of $\mathrm{CT}_{S,T}=\left( \mathrm{CH}_{\textnormal{ABE}}, \mathrm{CH}_{\textnormal{SUE}}, C \right)$ with the time updated by a unit.
	We fix $\mathrm{CH}^{\prime}_{\textnormal{SUE}}$ by running \linebreak \algorithm{SUE}{Update$\mathrm{CT}$}
	$(\mathit{\mathrm{PI}},\mathit{\mathrm{PK}}_{\textnormal{SUE}}, \mathrm{CH}_{\textnormal{SUE}})$.
	The output ciphertext is:
	\[
		\mathrm{CT}_{S,T+1}=\left( \mathrm{CH}_{\textnormal{ABE}},\mathrm{CH}^{\prime}_{\textnormal{SUE}}, C\right).
	\]
	
	\end{description}
	
	\subsection{Correctness}\label{subsec:cor}
	The decryption process finds a match for the private set and the partition of the non-revoked users, which always exists when the user is not revoked, from the construction of the CS scheme~\cite{naor2001SC}. Then we use the keys corresponding to the found match to retrieve the session keys of SUE and ABE. Due to the choice of the private key when creating the keys $\mathit{\mathrm{SK}}_{\mathbb{A},u}$ and $\mathit{\mathrm{TK}}_{T,R}$, from the correctness of the ABE and SUE schemes we obtain $\mathrm{EK}_{\textnormal{ABE},k}=e(g,g)^{\gamma_{k}s}$ and $\mathrm{EK}_{\textnormal{SUE},k}=e(g,g)^{(\alpha - \gamma_{k})s}$.
	The decryption process can be expanded as:
	\begin{align*}
		C\cdot (\mathrm{EK}_{\textnormal{ABE},k}\cdot \mathrm{EK}_{\textnormal{SUE},k})^{-1}&= M \cdot e(g,g)^{\alpha s} \cdot (e(g,g)^{\gamma_{k}s} \cdot e(g,g)^{(\alpha - \gamma_{k})s})^{-1}\\
		&= M \cdot e(g,g)^{\alpha s} \cdot (e(g,g)^{\alpha s})^{-1}\\
		&=M.
	\end{align*}
	The time is correctly updated by \algorithm{RS-ABE}{Update$\mathrm{CT}$}, since we are using the update properties of the SUE scheme.
	
	\subsection{Security for a chosen-plaintext attack}
	\label{security_section}
	We are now ready to state the following result on the security of the key-policy RS-ABE scheme we have just described.
	
	\begin{theorem}[Security]
	\label{RS-ABE_security_th}
	If Assumptions~\ref{a1}, \ref{a2}, and \ref{a3} are valid, then the key-policy RS-ABE scheme is secure under a chosen-plaintext attack.
	\end{theorem}
	

The idea behind our proof starts from the strategy adopted by Lee et al.~\cite{lee2013RSABE} for the proof of security of their ciphertext-policy RS-ABE scheme. 
For our key-policy RS-ABE scheme we want to show that a PPT adversary $\mathcal{A}$ plays the security game with negligible advantage.
	
	First we are going to define the so-called \emph{semi-functional} algorithm, a slightly modified version of the algorithm of the existing schemes. Then we define indexes to identify each request of the adversary, and later we describe the game structure and the reductions to our assumptions.
	
	We use hybrid games to prove our result: we split the proof of security of the original game in a sequence of lesser proof involving \emph{hybrid games}, where a hybrid game and the following differ 	only slightly. These hybrid games are constructed by taking each single request of the adversary in account: starting from the base security game, each ensuing game keeps each request equal to its following one, except for a single request. During this request, the key supplied to the adversary changes, with respect to the previous game, from standard to semi-functional.

	\subsubsection{Semi-functional algorithms}
	\label{semifunctional.def}
	Here we define the semi-functional versions of the algorithm \algorithm{RS-ABE}{GenKey}, and \algorithm{RS-ABE}{Enrypt}. 
	In the proof of Theorem \ref{RS-ABE_security_th} (Section \ref{sec:proof}), we will also use the semi-functional version \algorithm{RS-ABE}{UpdateKeySF} of the algorithm \algorithm{RS-ABE}{UpdateKey} but we omit here the definition, since the SUE part of the scheme is only modified, for which we refer to the work done by Lee et al.~\cite{lee2013RSABE}.  
	Henceforth we refer to the unmodified version of those algorithms as \emph{standard}. In the hybrid games the standard version of each algorithm will be gradually swapped by its semi-functional counterpart.

	For all semi-functional algorithms we fix a generator $g_2$ of $\G_{p_2}$. Then we fix randomly in $\Z_N$ two elements $\zeta_i$, $\eta_i$ for each node $\nu_i$ of the binary tree $\tree$. These elements will be used to link each other the semi-functional keys between multiple requests, $\zeta_i$ will be associated with semi-functional private key generation, and $\eta_i$ will be associated with time-update key.	
Moreover, for each possible attribute in $\mathfrak{A}$ we fix an element $z_i$.

	\begin{description}
	
\algorithmdefsingleoutput{RS-ABE}{GenKeySF}{\mathrm{PI},\mathrm{PK},\mathrm{MK},\mathbb{A},u}{\mathrm{SK}_{\mathbb{A},u}} first we generate a standard private key $\mathrm{SK}^{\prime}_{\mathbb{A},u}=( \mathrm{PV}_u, \mathrm{SK}^{\prime}_{\textnormal{ABE},k_1}, \ldots, \mathrm{SK}^{\prime}_{\textnormal{ABE},k_d} )$ using the standard version of the algorithm, \algorithm{RS-ABE}{GenKey}, considering as parameter the user $u$ and the access structure $\mathbb{A}=(B,\rho)$, $B\in\matrixset{\varRow}{\varColumn}{\Z_N}$. The private set is the set $\mathrm{PV}_u = \{U_{k_1},\ldots,U_{k_d}\}$.
	
	We use a construction similar to the one employed in \algorithm{ABE}{GenKey}. For each ABE key $\mathrm{SK}^\prime_{\textnormal{ABE},k_h}=(\{K^\prime_{1,j},K^\prime_{2,j}\}_{j=1}^{\varRow})$ associated with the set $U_{k_h}\in \mathrm{PV}_u$ we fix some random elements $r^\prime_{2,h},\ldots,r^\prime_{\varColumn,h}\in \Z_N$ and we set $\pvec[h]{v}=(\zeta_{k_h},r^\prime_{2,h},\ldots,r^\prime_{\varColumn,h})$. 
	Then the corresponding semi-functional ABE key is given by \linebreak  $
	\mathrm{SK}_{\textnormal{ABE},k_h}=
	\left(\left\{
	K_{1,i}=K^\prime_{1,i}g_2^{B_i\cdot \pvec[h]{v}},
	K_{2,i}=K^\prime_{2,i}\right
	\}_{i=1}^{\varRow}\right).
	$
	
	The output is then:
	\[
		\mathrm{SK}_{\mathbb{A},u}=\left( \mathrm{PV}_u, \mathrm{SK}_{\textnormal{ABE},k_1}, \ldots, \mathrm{SK}_{\textnormal{ABE},k_d} \right).
	\]

\algorithmdef{SUE}{EncryptSF}{\mathrm{PI},\mathrm{PK},T,M,c}{C,\mathrm{CH}_{T}} we omit the definition of this  function, since the SUE part of the scheme is only modified, for which we refer to the work done by Lee et al.~\cite{lee2013RSABE}.

\algorithmdef{ABE}{EncryptSF}{\mathrm{PI},\mathrm{PK},S,M,c}{C,\mathrm{CH}_{S}} first we generate a standard ciphertext header $\mathrm{CH}^\prime_{S}=\left(C^\prime_0,\left\{C^\prime_{1,a}\right\}_{a\in S} \right)$ associated with a set of attributes $S$ and a standard session key $\mathrm{EK}^\prime$.
	
	Then the output is: \[C=M\cdot\mathrm{EK}^\prime,\] and \[\mathrm{CH}_{S}=\left(C_0=C^\prime_0 g_2^c,\left\{C_{1,a}=C^\prime_{1,a} g_2^{cz_a} \right\}_{a\in S} \right).\]

	\algorithmdefsingleoutput{RS-ABE}{EncryptSF}{\mathrm{PI},\mathrm{PK},M,S,T}{\mathrm{CT}_{S,T}} first we generate a standard ciphertext $\mathrm{CT}^{\prime}_{S,T}=\left( \mathrm{CH}^\prime_{\textnormal{ABE}},\mathrm{CH}^\prime_{\textnormal{SUE}},C^\prime \right)$ using \algorithm{RS-ABE}{Encrypt}, associating it to the set of attributes $S$ and the time $T$.
	Now we create a semi-functional header both for the ABE part and the SUE part of the scheme.
	We choose randomly the element $c\in \Z_N$, and	
	we input $c$ in the previous functions: we run \algorithm{ABE}{EncryptSF}$(\mathrm{PI},\mathrm{PK}_{\textnormal{ABE}},S,M,c)$ and \algorithm{SUE}{EncryptSF}$(\mathrm{PI},\mathrm{PK}_{\textnormal{SUE}},T,M,c)$ in order to obtain $\mathrm{CH}_{\textnormal{ABE}}$ and $\mathrm{CH}_{\textnormal{SUE}}$, discarding the session keys.
	
	The output ciphertext is:
	\[
		\mathrm{CT}_{S,T}=\left( \mathrm{CH}_{\textnormal{ABE}},\mathrm{CH}_{\textnormal{SUE}},C=C^\prime \right).
	\]
	
	\end{description}
	
	
	Observe that knowing the semi-functional private key and update key does not allow us to decrypt a semi-functional encrypted text. The result of a decryption attempt would result in a randomized plaintext.

	\section{Proof of Security}\label{sec:proof}
	\subsection{Key indexing}
	During the proof of Theorem~\ref{RS-ABE_security_th} we want to be able to easily manage each different key request made in the querying phases of the security game.
	Each RS-ABE private key is a list of ABE keys, each of them associated with a different node of the tree $\tree$. The same is true for a time-update key, which is basically a list of SUE private keys.
	This means that the adversary makes multiple requests of private ABE and SUE keys associated with the same node.
	
	We consider for example the requests for ABE keys; the notation is identical for SUE keys. We are going to identify each of them using a pair of integers $(i_n,i_c)$. The first component is called \emph{node index}, and the second component is called \emph{counter index}.
	
	Their value is assigned as follows: suppose that a request is made for an ABE key associated with the node $\nu$.
	\begin{itemize}
		\item
		The node index $i_n$ is computed in two different ways, depending on the number of time the same node was requested.
		\begin{itemize}
			\item
			If this is the first time at which a key associated with the node $\nu$ was requested, then we set $i_n$ equal to the number of distinct nodes associated with previous ABE requests.
			\item
			Otherwise, we set the $i_n$ equal to the value of the node index of the previous request for the same node $\nu$.
		\end{itemize}
		In a request for the same node the node index remains always the same.
		\item
		The counter index $i_c$ is instead incremental for the same node, which means the following.
		\begin{itemize}
			\item
			If this is the first time at which a key associated with the node $\nu$ was requested, then we set $i_c$ equal to $1$.
			\item
			Otherwise, we consider the counter index of the previous node and we set $i_c$ as one more of that value.
		\end{itemize}
		In particular, when the same node is requested, the counter index is always different and increasing.
	\end{itemize}
	
	\subsection{Proof}
	\begin{proof}[Theorem \ref{RS-ABE_security_th}]
		We start by defining the hybrid games used in the proof. 
		
		\gamedef{G}{0}{
		The game \game{G}{0} is the standard security game, defined in Section~\ref{security.sec}: the private keys and time-update keys are standard, as is the challenge ciphertext.
		}
		
		\gamedef{G}{1}{
		This game is almost equal to \game{G}{0}, with the exception of the challenge ciphertext, which is semi-functional, i.e., computed using \algorithm{RS-ABE}{EncryptSF} instead of \algorithm{RS-ABE}{Encrypt}.
		}
		
 		\gamedef{G}{1,h}{
 		The index $h$ ranges from $0$ to $q_n$, where $q_n$ is the total number of distinct nodes of $\tree$ for which the adversary can query the corresponding ABE private key or SUE private key.
 		In the game \game{G}{1,h} the challenge ciphertext is semi-functional, and the keys are given as follow.
 			If the request for the ABE or SUE key is identified by a node index $i_n$ less or equal than $h$, then the key supplied by the challenger is semi-functional.
 			Otherwise, the key is a standard ABE or SUE key.
		}
		
		\gamedef{G}{2}{
		In this game everything is semi-functional: the challenge ciphertext, all the private keys and all the time-update keys.
		}
		
		\gamedef{G}{3}{
		The last game is equal to \game{G}{2}, but the ciphertext is random, which means that if 
		$\mathrm{CT}^\prime_{S,T}=( \mathrm{CH}_{\textnormal{ABE}},\mathrm{CH}_{\textnormal{SUE}},C=\Omega^s \cdot M ) $ is the semi-functional ciphertext for the plaintext $M$ in game \game{G}{2}, the ciphertext of \game{G}{3} is 
		$
		\mathrm{CT}^\prime_{S,T}=\left(\mathrm{CH}_{\textnormal{ABE}},\mathrm{CH}_{\textnormal{SUE}},C=\Omega^t \cdot M\right),
		$
		where $t$ is a random element in $\Z_N$.
		}
		
		In particular we can notice how the game \game{G}{1,0} has no semi-functional keys, because the node index is always greater than zero, meaning that $\mathgame{G}{1,0}=\mathgame{G}{1}$. Similarly for \game{G}{1,q_n}: each node index is less or equal than $q_n$, hence $\mathgame{G}{1,q_n}=\mathgame{G}{2}$.
	
	\vspace{-1em} 
	\begin{table}[H]
		\begin{center}
		\makebox[\textwidth][c]{ 
		\renewcommand{\arraystretch}{1.2}

		\begin{tabular}{M{1cm}| M{1.1cm} M{1.1cm} M{1.1cm} | M{1.1cm} M{1.1cm} M{1.1cm} | M{2.5cm} N }
				&	\multicolumn{3}{c|}{\textbf{Private key}}	&	\multicolumn{3}{c|}{\textbf{Time-update key}}	&	\textbf{Ciphertext} &	\\
				&	\multicolumn{3}{c|}{\textbf{$\mathrm{SK}_{\mathbb{A},u}$ }}	&	\multicolumn{3}{c|}{\textbf{$\mathrm{TK}_{T,R}$}}	&	\textbf{$\mathrm{CT}_{S,T}$} &	\\
			
			\hline 
			\game{G}{0}	&	\multicolumn{3}{c|}{\normal}	&	\multicolumn{3}{c|}{\normal}	&	\normal	&	\\
			\game{G}{1}	&	\multicolumn{3}{c|}{\normal}	&	\multicolumn{3}{c|}{\normal}	&	\semifunctional	&	\\
			\multirow{2}{*}{\game{G}{1,h}}	&	$i_n < h$ & $i_n = h$ &	$i_n > h$	&	$i_n < h$ & $i_n = h$ &	$i_n > h$	&	\multirow{2}{*}{\semifunctional}	&	\\ 
			[-0.5em]
				&	\semifunctional&\semifunctional&\normal	&	\semifunctional&\semifunctional&\normal	&		&	\\ 
			\game{G}{2}	&	\multicolumn{3}{c|}{\semifunctional}	&	\multicolumn{3}{c|}{\semifunctional}	&	\semifunctional	&	\\
			\game{G}{3}	&	\multicolumn{3}{c|}{\semifunctional}	&	\multicolumn{3}{c|}{\semifunctional}	&	\random	&	\\
		\end{tabular}
		}
		\end{center}
		\caption{The structure of the components for each hybrid game; \normal{} means standard, \semifunctional{} means semi-functional and \random{} means random.}
	\end{table}
	
	We call $\adv_{\mathcal{A}}^{\textnormal{G}_\textnormal{i}}$ the advantage of the adversary $\mathcal{A}$ for the game \game{G}{i}.
	\game{G}{0} is the original game, therefore $\adv_{\mathcal{A}}^{\textnormal{G}_0}=\adv_{\mathcal{A}}^{\textnormal{RS-ABE}}$. Moreover, since the final game has a randomly chosen session key, we know that the advantage for the last game is $\adv_{\mathcal{A}}^{\textnormal{G}_3}=0$. We also use the fact that $\adv_{\mathcal{A}}^{\textnormal{G}_1}=\adv_{\mathcal{A}}^{\textnormal{G}_{1,0}}$ and $\adv_{\mathcal{A}}^{\textnormal{G}_2}=\adv_{\mathcal{A}}^{\textnormal{G}_{1,q_n}}$.
		
	We split the advantage of the RS-ABE scheme using the advantages on distinguishing hybrid games. This advantage will later be bounded by its advantage in solving the problem in assumptions \ref{a1}, \ref{a2}, and \ref{a3} by Lemma~\ref{lemma.g0}, \ref{lemma.g1h}, and \ref{lemma.g2}, which are used for the inequality in the last line.
	We keep the dependence on $\lambda$ implicit.
	\begin{align*}
		&\adv_{\mathcal{A}}^{\textnormal{RS-ABE}} =\\ 
		&=\adv_{\mathcal{A}}^{\textnormal{G}_0} + \left(\adv_{\mathcal{A}}^{\textnormal{G}_1}-\adv_{\mathcal{A}}^{\textnormal{G}_1}\right) + \sum_{h=1}^{q_n} \left( \adv_{\mathcal{A}}^{\textnormal{G}_{1,h}}-\adv_{\mathcal{A}}^{\textnormal{G}_{1,h}} \right) -\adv_{\mathcal{A}}^{\textnormal{G}_{3}} \\
		&\leq \left| \adv_{\mathcal{A}}^{\textnormal{G}_0} - \adv_{\mathcal{A}}^{\textnormal{G}_1}\right| + \sum_{h=1}^{q_n} \left| \adv_{\mathcal{A}}^{\textnormal{G}_{1,h-1}}-\adv_{\mathcal{A}}^{\textnormal{G}_{1,h}} \right| + \left| \adv_{\mathcal{A}}^{\textnormal{G}_2} - \adv_{\mathcal{A}}^{\textnormal{G}_3}\right| \\
		&\leq \adv_{\mathscr{B}_1}^{\textnormal{A1}} + \sum_{h=1}^{q_n} 4\left(\left(q_{h}^\textnormal{ABE}+q_{h}^\textnormal{SUE}\right)+1\right)\adv_{\mathscr{B}_2}^{\textnormal{A2}} + \adv_{\mathscr{B}_3}^{\textnormal{A3}}.
	\end{align*}
	
	The number $q_{h}^{\textnormal{ABE}}$ (respectively $q_{h}^{\textnormal{SUE}}$) is the number of ABE (respectively SUE) key requests for the node with node index $h$, i.e., the highest value of a counter index in a request with node index $h$.
	The advantages $\adv_{\mathscr{B}_1}^{\textnormal{A1}}$, $\adv_{\mathscr{B}_2}^{\textnormal{A2}}$, and $\adv_{\mathscr{B}_3}^{\textnormal{A3}}$ in the last line are the one for an adversary who solves the assumption A\ref{a1}, A\ref{a2}, A\ref{a3} respectively, with the algorithms $\mathscr{B}_1$, $\mathscr{B}_2$,  $\mathscr{B}_3$ defined in the corresponding lemma.
	
	We recall now the properties of the CS scheme (Section~\ref{cs.sec}) which estimates the number of elements in $\mathrm{PV}_u$ and $\mathrm{CV}_R$: the elements of $\mathrm{PV}_u$ are at most $\log_2 N_\textnormal{max}$, and the elements of $\mathrm{CV}_R$ are at most $r_\textnormal{max}\log_2\left({N_\textnormal{max}}/{r_\textnormal{max}}\right)$, where $r_\textnormal{max}$ is the maximum size for a set of revoked user. Calling $q_{\textnormal{sk}}$ the number of requests for RS-ABE private keys and $q_{\textnormal{tu}}$ the number of requests for time updates, we obtain that:
	\begin{align*}
		\sum_{h=1}^{q_n}\left(q_{h}^{\textnormal{ABE}}+q_{h}^{\textnormal{SUE}}\right) 
		&\leq q_{\textnormal{sk}}\log_2 N_\textnormal{max} + q_{\textnormal{tu}} r_\textnormal{max} \log_2\left( \frac{N_\textnormal{max}}{r_\textnormal{max}}\right) \\
		\leq& \left(q_{\textnormal{sk}}+q_{\textnormal{tu}}\right) r_\textnormal{max} \log_2 N_\textnormal{max}.
	\end{align*}
	
	Thus the final formula is:
	\begin{align*}
	\adv_{\mathcal{A}}^{\textnormal{RS-ABE}}(\lambda) \leq & \adv_{\mathscr{B}_1}^{\textnormal{A1}}(\lambda) + \bigO{q \cdot r_\textnormal{max} \cdot \log_2 N_\textnormal{max} }\adv_{\mathscr{B}_2}^{\textnormal{A2}}(\lambda) + \adv_{\mathscr{B}_3}^{\textnormal{A3}}(\lambda),
	\end{align*}
	where $q=q_{\textnormal{sk}}+q_{\textnormal{tu}}$ is the total number of requests from the adversary.
	\end{proof}
	Now we prove the three aforementioned lemmas.
	Each lemma uses only a single assumption: the order of presentation is based on the required assumption.
	For any detail regarding the SUE scheme we refer to the paper by Lee et al.~\cite{lee2013RSABE}, since the procedure in regard to the SUE scheme is the same.
	\subsubsection{Changing the ciphertext}
	\begin{lemma}
		\label{lemma.g0}
		If Assumption \ref{a1} is true, then no PPT adversary $\mathcal{A}$ is able to distinguish between~\game{G}{0} and~\game{G}{1}, which means that $\left| \adv_{\mathcal{A}}^{\mathgame{G}{0}} - \adv_{\mathcal{A}}^{\mathgame{G}{1}}\right|$ is negligible with respect to $\lambda$.
	\end{lemma}
	\begin{proof}
		We prove this lemma by contradiction, showing that if a PPT adversary $\mathcal{A}$ could distinguish between~\game{G}{0} and~\game{G}{1}, then there would exist a PPT algorithm $\mathscr{B}_1$ which breaks assumption~\ref{a1}.
	
	Let us take $D=\left(\mathscr{G}(\lambda),g_1,g_3\right)$ and $W$ as inputs, the latter either of the form~$W_1\in\G_{p_1}$ or~$W_2\in\G_{p_1p_2}$.
	We create a simulator $\mathscr{B}_1$ that from the input constructs two games, whose probability distribution is the same as random construction of the~\game{G}{0} (if $W=W_1$) or~\game{G}{1} (if $W=W_2$).
	
	\phase{Setup}
	We follow the procedure of \algorithm{RS-ABE}{Setup}, but when applying \algorithm{ABE}{Setup} and \algorithm{SUE}{Setup} we choose the random components of both of them by fixing a random exponent for the corresponding subgroup generator, instead of directly making the random choices between the elements of the subgroup.
	Let $l$ be the maximum length of a label string, obtained from the maximum time $T_\textnormal{max}$ as $l=d_\textnormal{max}=\left\lceil \log_2(T_\textnormal{max}+2) \right\rceil-1$.
	\randomchoose{p}{
		$a,\gamma\in\Z_N$;
		$t_j^\prime \in \Z_N\text{, for all }j\in \mathfrak{A}$.
	}

	We fix $g=g_1$. Then we set the ABE keys as:	
	\begin{align*}
		\mathrm{MK}_{\textnormal{ABE}}&=\left( \gamma, Z=g_3 \right);\\
		\mathrm{PK}_{\textnormal{ABE}}&=\left( \left(N,\G,\GT,e\right), g, \left\{g^{t_j}\right\}_{j\in \mathfrak{A}}, e(g,g)^\gamma \right).
	\end{align*}
	\disclaimerSUE{}
	By choosing random exponents we are choosing random elements of the corresponding groups with a uniform distribution, thus the output of this procedure has the same probability distribution of the corresponding outputs of the standard setup procedure.
	Moreover we create the binary tree $\tree$ associating to each node $\nu_i$ a random exponent $\gamma_{i}$.
	The master private key is then:
	\[
		\mathrm{MK}=\left(\alpha, \tree,  \mathrm{MK}_{\textnormal{ABE}}, \mathrm{MK}_{\textnormal{SUE}} \right).
	\]
	
	The public key, which is given to the distinguishing adversary $\mathcal{A}$, is then:
	\[
		\mathrm{PK}=\left( \mathrm{PK}_{\textnormal{ABE}}, \mathrm{PK}_{\textnormal{SUE}}, g, \Omega=e(g,g)^\alpha \right).
	\]
	
	\phase{Query 1}
	During this phase a standard private key or time-update key is provided to each request of $\mathcal{A}$, since we know everything needed in order to run \algorithm{RS-ABE}{GenKey} or \algorithm{RS-ABE}{UpdateKey}.
	
	\phase{Challenge}
	We receive from $\mathcal{A}$ the attributes $S_*$, the time $T_*$, and the messages $M_*^0$ and $M_*^1$ upon which the adversary wish to be challenged.
	
	A standard RS-ABE ciphertext has the form $\left( \mathrm{CH}_{\textnormal{ABE}},\mathrm{CH}_{\textnormal{SUE}},C=\Omega^s M \right)$, where $\mathrm{CH}_{\textnormal{ABE}}=\left(C_0=g^s,\left\{C_{1,i}=T_i^s\right\}_{i\in S} \right)$ and $\mathrm{CH}_{\textnormal{SUE}}$ is a SUE ciphertext. We are going to create a particular ciphertext which generates everything randomly as in the standard construction, but which sets $C_0=W$. If~$W$ is of the form $W_1$, this implies that the key is standard, otherwise it will be a semi-functional ciphertext, thus exploiting $\mathcal{A}$ gives us an advantage in distinguishing $W$.
	
	We show how to create ABE ciphertext:
	\[
		\mathrm{CH}_{\textnormal{ABE}}=\left(C_0=W,\left\{C_{1,i}=W^{t_i}\right\}_{i\in S} \right).
	\]
	
	Observe that if $W=W_1=g_1^s\in\G_{p_1}$, then the above header is exactly a standard ciphertext header for the ABE scheme, since $W^{t_i}=g_1^{st_i}=T_i^s$.
	Otherwise we know that $W$ can be written as a product of an element of $\G_{p_1}$ and an element of $\G_{p_2}$: $W=W_2=g_1^sg_2^m$, for some unknown generator $g_2$.
	Recalling that a semi-functional ciphertext header has the form 
	\[
		\mathrm{CH}_{\textnormal{ABE}}=\left(C_0=C^\prime_0 g_2^c,\left\{C_{1,i}=C^\prime_{1,i} g_2^{cz_i} \right\}_{i\in S} \right),
	\]
	we can observe that the case $W=W_2$ corresponds to a semi-functional ciphertext header with $c=m$ and $z_i=t_i$.
	
	However the $z_i$ are certainly not random and independent on the value of $t_i$, hence there might be differences between this key and a random key generated by \algorithm{ABE}{EncryptSF}.
	This is not the case: the elements $t_i$ are randomly generated in $\Z_N$, but we are only interested in the \emph{$p_1$ part} of $t_i$, i.e., $t_i\mod p_1$. In fact $g=g_1$ is an element of $\G_{p_1}$, and thus has the form $h^{x p_2 p_3}$ for some generator $h$ of $\G$ and $x\in\Z_N$, and obviously $h^{p_1 p_2 p_3}=1$.
	On the other side, using the same argument we notice that we are not interested in all the information of $z_i$, but only on its $p_2$ part, since $g_2$ is a generator of $\G_{p_2}$.
	
	From the Chinese Remainder Theorem we know that, since $p_1$, $p_2$, and $p_3$ are three distinct primes, there is a bijection from $\Z_N$ and $\Z_{p_1}\times\Z_{p_2}\times\Z_{p_3}$ defined by $n\mapsto(n\mod p_1,n\mod p_2,n\mod p_3)$. Therefore choosing a random $t_j$ is equivalent to choosing randomly and independently each $p$ part.
	
	In conclusion, the final distribution for the challenge ciphertext header is the same as the output of \algorithm{RS-ABE}{Encrypt} if $W=W_1$ or \algorithm{RS-ABE}{EncryptSF} if $W=W_2$.
	
	\phase{Query 2}
	This phase is equal to the first querying phase.
	
	\phase{Guess}
	We obtain the guess $b$ from $\mathcal{A}$: this is the final output of $\mathscr{B}_1$.
	
	Since the adversary $\mathcal{A}$ has a non-negligible advantage on differentiating \game{G}{0} from \game{G}{1}, and this difference allows us to correctly guess the type of $W$, then the algorithm $\mathscr{B}_1$ has a non-negligible advantage, which is a contradiction.
	\end{proof}

	The proof structure we have just employed will recur many times when proving that the difference between two games is negligible: we suppose that we are able to distinguish the two, and from this hypothesis we will build a solver for the assumption that are involved.
	
	\subsubsection{Changing each key}
	\ \\
	The proof of the second lemma is more involved, since we need to take care of the ABE and SUE keys which compose each request for a private key or a time-update request.	
	For this reason the reduction is split in a finer sequences of games: each game transforms into the following by changing the behavior of requests associated only with a fixed counter index~$i_c$.
	
	We will consider two separate cases depending on the request type, respectively for an ABE or a SUE key. Let us fix any challenge attributes $S_*$ and challenge time $T_*$; with the querying phase restrictions of the security game defined in Section~\ref{security.sec} the adversary is able to perform only two kinds 
	 of requests for a fixed node $\nu$ with node index $h$: we can thus split the adversary into two possible types.
	\begin{definition}
		We say that an adversary $\mathcal{A}$ is of \emph{\typeone{}} if each access structure $\mathbb{A}$ of an ABE private key associated with the node $\nu$ does not contain the set~$S_*$.
		
		An adversary $\mathcal{A}$ is of \emph{\typetwo{}} if the update time $T$ of every SUE private key associated with the node $\nu$ is less than~$T_*$.
	\end{definition}
	
	This distinction is not a partitioning of all the possible adversaries, but each of them is described by at least one of these types. If we consider the description of the challenge phase, we observe that the two types correspond to the conditions $S_*\not\in\mathbb{A}$ and $T_*>T$ respectively. If by contradiction we suppose that the adversary is neither of \typeone{} nor of \typetwo{}, then there exists at least a valid ABE key and a valid SUE key for the challenge time and set of attributes. Since our security game does not allow the requested keys to directly decipher the ciphertext, this means that the user who owns the secret key must be revoked, $u\in R$. However this is impossible, since we are considering only requests for the node $\nu$: in order to be a node involved in an ABE request, the adversary must require a private key for a user whose position on the binary tree $\tree$ is below the node $\nu$. But if we revoke any one single user $u$ below the node $\nu$, then the node $\nu$ would belong to the path from $u$ to the root node, and thus it would belong to the Steiner tree of the set of revoked users. This means that any valid time-update key from a time-update request would not contain a SUE key associated with the node $\nu$.
	
	We are now ready to prove the following.
	
	\begin{lemma}
		\label{lemma.g1h}
		If Assumption \ref{a2} is true, then no PPT adversary $\mathcal{A}$ is able to distinguish between \game{G}{1,h-1} and \game{G}{1,h}, which means that $\left| \adv_{\mathcal{A}}^{\mathgame{G}{1,h}} - \adv_{\mathcal{A}}^{\mathgame{G}{1,h-1}}\right|$ is negligible with respect to $\lambda$.
	\end{lemma}
	\begin{proof}
		We consider first an adversary of \typeone{} which makes at most $q_{h}^\textnormal{ABE}$ ABE key requests for the node with node index $h$, and we define the corresponding hybrid games. Each of them is based on the structure of the game \game{G}{1,h-1}, i.e., a game where the requested keys for both ABE and SUE are semi-functional for the nodes with node index less than~$h$, standard for a bigger node index, and the ciphertext is semi-functional. The only difference between \game{G}{1,h-1} and \game{G}{1,h} occurs during the requests for ABE keys associated with nodes with node index exactly equal to $h$. We modify the behavior of the games basing on the value of the counter index of the node,~$i_c$. In the following we assume that the index~$k$ ranges from $0$ to $q_{h}^\textnormal{ABE}$.
		
		\gamedef{H}{k,1}{
		We define the behavior for the $i$\ith{} key request, associated with a node index $i_n=h$ basing on its counter index $i_c$.
		
		\begin{itemize}
			\item
			If $i_c=k$ we generate a standard ABE key $\mathrm{SK}^\prime=(\left\{K_{1,i},\allowbreak K_{2,i}\right\}_{i=1}^\varRow )$ as previously.
			\randomchoose{p}{%
				$u_j\in\Z_N\text{, for all }j\in \{1,\ldots,\varColumn\}$,
				$\delta_i\in\Z_N\text{, for all }i\in \{1,\ldots,\varRow\}$.
			}
			We give to the adversary the following semi-functional key:
			\[
				\mathrm{SK}_\textnormal{ABE}=\left( \left\{K_{1,i}\cdot g_2^{B_i \cdot \vec{u}+\delta_i z_{\rho(i)}}, K_{2,i}\cdot g_2^{\delta_i}\right\}_{i=1}^\varRow \right),
			\]
			where the $z_i$'s are the elements randomly chosen in order to create the ABE semi-functional ciphertext.
			We call this kind private keys \emph{semi-functional of type 1} (\emph{\sfone}, as suggested in~\cite{lewko2010ABE}.
			
			\item
			If $i_c<k$, we create a standard ABE key $\mathrm{SK}^\prime=(\left\{K_{1,i},\allowbreak K_{2,i}\right\}_{i=1}^\varRow )$ using as parameters the access structure $\mathbb{A}$ and the user $u$ given as input by the adversary, where $\varRow$ is the number of rows and $\varColumn$ is the number of columns of the matrix $B$ defining the access structure $\mathbb{A}=(B,\rho)$.
			\randomchoose{p}{%
				$u_i\in\Z_N\text{, for all }i\in \{1,\ldots,\varColumn\}$.
			}
			Moreover we fix $\vec{u}=(u_1,\ldots,u_\varColumn)$, and then we output the semi-functional key defined by:
			\[
				\mathrm{SK}_\textnormal{ABE}=\left( \left\{K_{1,i}\cdot g_2^{B_i \cdot \vec{u}},K_{2,i}\right\}_{i=1}^\varRow \right).
			\]
			Those keys are called \emph{semi-functional of type 2} (\emph{\sftwo}).
			
			\item
			If $i_c>k$ then the challenger gives to the adversary a standard ABE private key.
		\end{itemize}
		The two types of semi-functional keys are different from the semi-functional ABE keys which we have previously defined in \ref{semifunctional.def}: the vector $\vec{u}$ is generated randomly each time for every new key request, but for a semi-functional key the first component of the vector is fixed at the start, depending on the node associated with the request.
		}
		
		\gamedef{H}{k,2}{
		This game is equal to the previous game, with the exception that the case $i_n=h$ and $i_c=k$ gives a semi-functional key of type 2, or equivalently that the elements $\delta_i$ are equal to zero.
		}
		
		\gamedef{H^\prime}{k,1}{
		This game is equal to the game \game{H}{k,1}, with the exception that the cases with node index $i_n=h$ and counter index $i_c\geq k$ give a different semi-functional ABE key. We first generate a key $\mathrm{SK}^\prime=(\left\{K^{\prime}_{1,i},\allowbreak K^{\prime}_{2,i}\right\}_{i=1}^\varRow )$ with the same procedure as described in game \game{H}{k,1}, and then we add the contribution of the elements $\zeta_i$'s, defined in Section~\ref{semifunctional.def} to create the semi-functional keys.
		\randomchoose{p}{%
			$r_2,\ldots,r_\varColumn\in\Z_N$.
		}
		We take the element $\zeta_j$, where $j$ is the index of the node $\nu_j$ associated with the node index $h$, and we fix the vector $\vec{v}=(\zeta_j,r_2,\ldots,r_\varColumn)$. The output is:
		\[
			\mathrm{SK}_{\textnormal{ABE}}=
			\left(\left\{
			K^{\prime}_{1,i}\cdot g_2^{B_i\cdot \vec{v}},
			K^{\prime}_{2,i}\right
			\}_{i=1}^{\varRow}\right).
		\]
		We call \emph{semi-functional of type 3} (\emph{\sfthree}) the key resulting for the case ${i_c=k}$.
		
		We observe that if the starting key is standard, i.e., $i_c>k$, then the resulting key is semi-functional. In fact the added element does not change the distribution of the output: the combined exponent of the element $g_2$ is $B_i\cdot \vec{v} + B_i \cdot \vec{u}=B_i\cdot (\vec{v}+\vec{u})$, where $\vec{v}$ is the vector we have previously defined, and $\vec{u}$ is the random vector used in game \game{H}{k,1}. In fact the first component of the sum of the two vectors is the sum of a random element and a constant, and the other components are the sum of two random elements in $\Z_N$.
		Similarly, if we start with a semi-functional key of type one we obtain a key which has the same distribution as a randomly-generated semi-functional key of type one.
		}
		
		\gamedef{H^\prime}{k,2}{
		This game merges the behaviors of the games \game{H}{k,2} and \game{H^\prime}{k,1}. Again, this game is equal to the game \game{H}{k,2}, with the exception that the cases $i_n=h$ and $i_c\geq k$ give a modified ABE key. The modification is the same as in game \game{H^\prime}{k,1}, which means that the keys for $i_c > k$ are semi-functional, and the key for $i_c=k$, which we denote \emph{semi-functional of type 4} (\emph{\sffour}), has the same distribution of the key given in \game{H}{k,2}, that is of a semi-functional key of type 2.
		}
		
		\gamedef{H^{\prime\prime}}{}{
		This game is equal to the game \game{G}{1,h}.
		}
		
		\begin{remark}The only difference between semi-functional keys and semi-functional keys of type 2 is that the latter uses random elements in lieu of the elements $\zeta_j$ and $\eta_j$.
		\end{remark}
		
		We observe that the two games \game{H}{q_h^\textbf{ABE},2} and \game{H^\prime}{q_h^\textbf{ABE},2} are equal, since the output distribution for the only different case $i_n=h$, $i_c=q_h^\textnormal{ABE}$ is the same.
		The games \game{H^\prime}{0,1} and \game{H^\prime}{0,2} are equal: each ABE key associated with the node with node index $h$ is semi-functional. The difference between them and \game{H^{\prime\prime}}{} are only the SUE key requests with node index $h$: in the formers they are standard, but in the latter they are semi-functional.
		Moreover the games \game{H}{0,1} and \game{H}{0,2} are equal to \game{G}{1,h-1}.
		
		\begin{table}[H]
		\begin{center}
		\makebox[\textwidth][c]{ 
		\renewcommand{\arraystretch}{1.2}

		\begin{tabular}{M{1cm}| M{1.3cm} M{1.3cm} M{1.3cm} | M{0.5cm} M{0.5cm} M{0.5cm} !{\color{white}} N } 
				&	\multicolumn{3}{c|}{\textbf{ABE keys}}	&	\multicolumn{3}{c}{\textbf{SUE keys}}	&	\\
				&	\multicolumn{3}{c|}{\textbf{$\mathrm{SK}_{\textnormal{ABE}}$}}	&	\multicolumn{3}{c}{\textbf{$\mathrm{SK}_{\textnormal{SUE}}$}}	&	\\
			\hline
			\multirow{2}{*}{\game{H}{k,1}}	&	$i_c < k$ & $i_c = k$ &	$i_c > k$	&	\multicolumn{3}{c}{\multirow{2}{*}{\normal}}	&	\\ 
			[-0.5em]
				&	\sftwo&\sfone&\normal	&	&&	&	\\
			\multirow{2}{*}{\game{H}{k,2}}	&	$i_c < k$ & $i_c = k$ &	$i_c > k$	&	\multicolumn{3}{c}{\multirow{2}{*}{\normal}}	&	\\ 
			[-0.5em]
				&	\sftwo&\sftwo&\normal	&	&&	&	\\
			\multirow{2}{*}{\game{H^\prime}{k,1}}	&	$i_c < k$ & $i_c = k$ &	$i_c > k$	&	\multicolumn{3}{c}{\multirow{2}{*}{\normal}}	&	\\ 
			[-0.5em]
				&	\sftwo&\sfthree&\semifunctional	&	&&	&	\\
			\multirow{2}{*}{\game{H^\prime}{k,2}}	&	$i_c < k$ & $i_c = k$ &	$i_c > k$	&	\multicolumn{3}{c}{\multirow{2}{*}{\normal}}	&	\\ 
			[-0.5em]
				&	\sftwo&\sffour&\semifunctional	&	&&	&	\\
			\multirow{2}{*}{\game{H^{\prime\prime}}{}}	&	\multicolumn{3}{c|}{\multirow{2}{*}{\semifunctional}}	&	\multicolumn{3}{c}{\multirow{2}{*}{\semifunctional}}	&	\\ 
			[-0.5em]
				&		&		&	\\
			
		\end{tabular}
		}
		\end{center}
		\caption{The behavior of the security games for all requests with node index~$i_n=h$; the remaining requests are made as in \game{G}{1,h-1}.}
		
	\end{table}
	
		We call $\adv{}_{\mathcal{A}}^{\mathgame{H}{}}$ the advantage of the adversary $\mathcal{A}$ in a game \game{H}{}. In propositions \ref{prop.h1}, \ref{prop.h2}, \ref{prop.h3}, \ref{prop.h4}, and \ref{prop.h5} we will prove that the advantage in distinguishing the pair of games
			(\game{H}{k,1}, \game{H}{k-1,2}),
			(\game{H}{k,1}, \game{H}{k,2}),
			(\game{H^\prime}{k,1}, \game{H^\prime}{k-1,2}),
			(\game{H^\prime}{k,1}, \game{H^\prime}{k,2}), and
			(\game{H^\prime}{0,2}, \game{H^{\prime\prime}}{})
		is negligible with respect to the security parameter $\lambda$.
		Thus we split the advantage in distinguishing between \game{G}{1,h-1} and \game{G}{1,h} for an adversary $\mathcal{A}_1$ of \typeone{} with the following inequalities:
		
		\begin{align}
				\displaybreak[3]
			&\adv{}_{\mathcal{A}_1}^\mathgame{G}{1,h-1}-\adv{}_{\mathcal{A}_1}^\mathgame{G}{1,h} =
			\adv{}_{\mathcal{A}_1}^\mathgame{H}{0,2}-\adv{}_{\mathcal{A}_1}^\mathgame{H^{\prime\prime}}{} =
			\notag{} \\
				\displaybreak[3]
			&=\adv{}_{\mathcal{A}_1}^\mathgame{H}{0,2} + 
			\sum_{k=1}^{q_h^\textnormal{ABE}}
				\left( \adv{}_{\mathcal{A}_1}^\mathgame{H}{k,1}-\adv{}_{\mathcal{A}_1}^\mathgame{H}{k,1} \right) +
			\sum_{k=1}^{q_h^\textnormal{ABE}}
				\left( \adv{}_{\mathcal{A}_1}^\mathgame{H}{k,2}-\adv{}_{\mathcal{A}_1}^\mathgame{H}{k,2} \right) +
			\notag{} \\
				\displaybreak[3]
			&+\sum_{k=1}^{q_h^\textnormal{ABE}}
				\left( \adv{}_{\mathcal{A}_1}^\mathgame{H^\prime}{k,1}-\adv{}_{\mathcal{A}_1}^\mathgame{H^\prime}{k,1} \right) +
			\sum_{k=0}^{q_h^\textnormal{ABE}-1}
				\left( \adv{}_{\mathcal{A}_1}^\mathgame{H^\prime}{k,2}-\adv{}_{\mathcal{A}_1}^\mathgame{H^\prime}{k,2} \right) -
				\adv{}_{\mathcal{A}_1}^\mathgame{H^{\prime\prime}}{}
			\notag{} \\
				\displaybreak[3]
			&\leq\sum_{k=1}^{q_h^\textnormal{ABE}}
				\left| \adv{}_{\mathcal{A}_1}^\mathgame{H}{k-1,2}-\adv{}_{\mathcal{A}_1}^\mathgame{H}{k,1} \right| +
			\sum_{k=1}^{q_h^\textnormal{ABE}}
				\left| \adv{}_{\mathcal{A}_1}^\mathgame{H}{k,1}-\adv{}_{\mathcal{A}_1}^\mathgame{H}{k,2} \right| +
			\notag{} \\
				\displaybreak[3]
				\allowdisplaybreaks
			&+\sum_{k=1}^{q_h^\textnormal{ABE}}
				\left| \adv{}_{\mathcal{A}_1}^\mathgame{H^\prime}{k,2}-\adv{}_{\mathcal{A}_1}^\mathgame{H^\prime}{k,1} \right| +
			\sum_{k=1}^{q_h^\textnormal{ABE}}
				\left| \adv{}_{\mathcal{A}_1}^\mathgame{H^\prime}{k,1}-\adv{}_{\mathcal{A}_1}^\mathgame{H^\prime}{k-1,2} \right| +
			\notag{} \\
				\displaybreak[3]
				\allowdisplaybreaks
				\label{inequality.ABE}
				&+\left| \adv{}_{\mathcal{A}_1}^\mathgame{H^{\prime}}{0,2} - \adv{}_{\mathcal{A}_1}^\mathgame{H^{\prime\prime}}{} \right| \leq \left( 4q_h^\textnormal{ABE}+1\right)\adv_{\mathscr{B}_2}^{\textnormal{A2}}.
		\end{align}

		A similar splitting can be obtained considering an adversary of \typetwo{} who we suppose making at most $q_{h}^\textnormal{SUE}$ SUE key requests for the nodes with node index $h$.		
		The full description for games involving an adversary of \typetwo{}, and the corresponding proof for the inequality that follows, is found in Lee et al.~\cite{lee2013RSABE}.		
		The advantage for an adversary $\mathcal{A}_2$ of \typetwo{} in distinguishing between \game{G}{1,h-1} and \game{G}{1,h} satisfies:
				\begin{equation}
			\label{inequality.SUE}
			\adv{}_{\mathcal{A}_2}^\mathgame{G}{1,h-1}-\adv{}_{\mathcal{A}_2}^\mathgame{G}{1,h} \leq \left( 4q_h^\textnormal{SUE}+1\right)\adv_{\mathscr{B}_2}^{\textnormal{A2}}.
		\end{equation}
		
		Now we combine both adversary types by considering both the events \textit{the adversary is of \typeone{}}, $E_1$, and \textit{the adversary is of \typetwo{}}, $E_2$. If the adversary is both of \typeone{} and of \typetwo{} we can arbitrarily choose one of the two types, e.g.\@ we can assume that they are of \typeone{}.
		
		If the adversary is of \typeone{} we use inequality \ref{inequality.ABE}, otherwise we use \ref{inequality.SUE}. In particular, with this setting we know than there is no difference between the games \game{G}{1,h-1} and \game{G}{1,h} for both an adversary of \typeone{} and of \typetwo{}.
		Thus we can now conclude the proof with the following inequality:
		\begin{align*}
			&\adv{}_\mathcal{A}^\mathgame{G}{1,h-1}-\adv{}_\mathcal{A}^\mathgame{G}{1,h} = \\
			&= \prob{E_1}\left(\adv{}_{\mathcal{A}_1}^\mathgame{G}{1,h-1}-\adv{}_{\mathcal{A}_1}^\mathgame{G}{1,h}\right) + \prob{E_2}\left(\adv{}_{\mathcal{A}_2}^\mathgame{G}{1,h-1}-\adv{}_{\mathcal{A}_2}^\mathgame{G}{1,h}\right) \\
			&= \prob{E_1}\left(\adv{}_{\mathcal{A}_1}^\mathgame{H}{0,2}-\adv{}_{\mathcal{A}_1}^\mathgame{H^{\prime\prime}}{}\right) + (1-\prob{E_1})\left(\adv{}_{\mathcal{A}_2}^\mathgame{G}{1,h-1}-\adv{}_{\mathcal{A}_2}^\mathgame{G}{1,h}\right)\\
			&\leq \prob{E_1}\left(4q_h^\textnormal{ABE}+1\right)\adv_{\mathscr{B}_2}^{\textnormal{A2}} + (1-\prob{E_1})\left(4q_h^\textnormal{SUE}+1\right)\adv_{\mathscr{B}_2}^{\textnormal{A2}} \\	
			&\leq \left(4(q_h^\textnormal{ABE}+q_h^\textnormal{SUE})+2\right)\adv_{\mathscr{B}_2}^{\textnormal{A2}} +4(q_h^\textnormal{ABE}-q_h^\textnormal{SUE})\prob{E_1}\adv_{\mathscr{B}_2}^{\textnormal{A2}}\\
			&\leq \left(4(q_h^\textnormal{ABE}+q_h^\textnormal{SUE})+2\right)\adv_{\mathscr{B}_2}^{\textnormal{A2}}.
		\end{align*}
		
		The last inequality is obtained by supposing $q_h^\textnormal{SUE}\geq q_h^\textnormal{ABE}$, and then removing the resulting negative term. If $q_h^\textnormal{SUE}< q_h^\textnormal{ABE}$ we follow the same procedure, but using $E_2$ as reference set, in particular $\prob{E_1}=1-\prob{E_2}$.
	\end{proof}
	
	As in Lemma~\ref{lemma.g0} we are going to show that if a PPT adversary were able to distinguish between the two considered games, then we would be able to exploit this advantage in order to find a PPT algorithm which breaks Assumption~\ref{a2}.
	The only differences with the Lemma are found in the querying phases, where the key constructions is different between each game.
		
		We are given the two input in Assumption~\ref{a2}, which are the tuple \linebreak$D=\left(\mathscr{S},g_1,g_3,X_1Y_1,Y_2Z_1\right)$, where $X_1\in\G_{p_1}$, $Y_1,Y_2\in\G_{p_2}$, and $Z_1\in\G_{p_3}$, and an element $W$ of the group.
		The group element $W$ can be with equal probability of two different forms: either randomly chosen among all elements of $\G$, or just among the element of $\G_{p_1p_3}$. Equivalently, it can be chosen as a product of random elements, either $W=X_2Y_3Z_3$ or $W=X_2Z_3$, where $X_2\in\G_{p_1}$, $Y_3\in\G_{p_2}$, and $Z_2\in\G_{p_3}$, everything chosen with uniform distribution.
		
	We use the inputs to build the following simulator.
	
	\phase{Setup}
	The setup phase generates the public and private parameters of the scheme through random choices of exponents for each group element.
	
	\randomchoose{p}{%
		$a,\gamma,\alpha\in\Z_N$;
		$t_j^\prime \in \Z_N\text{, for all }j\in \mathfrak{A}$.
	}
	
	Let $g=g_1$. We set the ABE and SUE keys as:
	\begin{align*}
		\mathrm{MK}_{\textnormal{ABE}}&=\left( \gamma, Z=g_3 \right);\\
		\mathrm{PK}_{\textnormal{ABE}}&=\left( \left(N,\G,\GT,e\right), g, \left\{g^{t^\prime_j}\right\}_{j\in \mathfrak{A}}, e(g,g)^\gamma \right).
	\end{align*}
	\disclaimerSUE{}
		This means that, maintaining the notation used in the previous scheme constructions, $T_i=g^{t^\prime_i}$, $w=g^{w^\prime}$, $u_{i,b}=g^{u^\prime_{i,b}}$, and $h_{i,b}=g^{h^\prime_{i,b}}$.
	
	Since we are choosing random exponents, when we consider their corresponding elements of the group we are choosing them with a uniform distribution, and thus the output of this procedure has the same probability distribution of the corresponding outputs of \algorithm{ABE}{Setup} and \algorithm{SUE}{Setup}.
	We create the binary tree $\tree$ associating to each node $\nu_i$ a random exponent~$\gamma_{i}$.
	
	The output, given to the adversary $\mathcal{A}$, is then:
	\begin{align*}
		\mathrm{PK}&=\left( \mathrm{PK}_{\textnormal{ABE}}, \mathrm{PK}_{\textnormal{SUE}}, g, \Omega \right).
	\end{align*}
	
	The secret key, which we fully know, is 
	\[
		\mathrm{MK}=\left( \mathrm{MK}_{\textnormal{ABE}}, \mathrm{MK}_{\textnormal{SUE}}, \alpha, \tree \right).
	\]
	
	We also generate randomly for each node $\nu_i$ the elements $\zeta_i,\eta_i\in\Z_N$, which are needed to create the semi-functional keys. We need also the elements $x_i,y_i\in\Z_N$, for each $i\in\{0,\ldots,l\}$, used in the construction of a SUE ciphertext.
	
	Exception will be Proposition~\ref{prop.h5}, where for a specific node $\nu_j$ the elements $\zeta_j$ and $\eta_j$ will be implicitly defined using input elements.
	
	\phase{Query 1}
	In this phase we give to the adversary $\mathcal{A}$ the required keys.
	
	First we observe that we are able to generate standard private and time-update keys, since we own all the needed elements of the private master key. We are also able to build semi-functional keys, even if we cannot generate them directly, because we do not know a generator for the subgroup $\G_{p_2}$. The following con\-struc\-tions leverage the element $Y_2Z_1$ given as input, and fix implicitly $g_2=Y_2$. 
	
	Suppose that the request is associated with the node $\nu_j$.
	A semi-functional ABE key can be obtained with the following procedure. First we build a standard ABE key, $\mathrm{SK}^\prime_{\textnormal{ABE}}=(\{K^\prime_{1,i},K^\prime_{2,i}\}_{i=1}^{\varRow})$, for the access structure $\mathbb{A}=(B,\rho)$, $B\in\matrixset{\varRow}{\varColumn}{\Z_N}$.
	We fix some random elements $r^\prime_{2},\ldots,r^\prime_{\varColumn}\in \Z_N$ and we set $\pvec{v}=(\zeta_{j},r^\prime_{2},\ldots,r^\prime_{\varColumn})$. 
	The semi-functional ABE key is:
	\[
		\mathrm{SK}_{\textnormal{ABE}}=
		\left(\left\{
		K_{1,i}=K^\prime_{1,i}(Y_2Z_1)^{B_i\cdot \pvec{v}},
		K_{2,i}=K^\prime_{2,i}\right
		\}_{i=1}^{\varRow}\right).
	\]
	The distribution is the same as a standard semi-functional key: the only difference is the presence of the factor $Z_1^{B_i\cdot \pvec{v}}\in\G_{p_3}$ in $K_{1,i}$, but since when we generate $K^\prime_{1,i}$ we use a random element of $\G_{p_3}$ as a factor, the contribution of this element is just a translation inside the space $\G_{p_3}$, which does not change the distribution of our random choice.
	We are also able to create ABE semi-functional keys of type 2, since its construction is the same apart from the elements $\zeta_j$ and $\eta_j$, which are swapped with random elements.
	
	We can also obtain semi-functional SUE key. Everything regarding the SUE part of the scheme is omitted, for more details we refer to Lee et al.~\cite{lee2013RSABE}.
	
	We must consider three different cases in the querying phase.
	\begin{itemize}
		\item
		If $i_n<h$ we output a semi-functional key.
		\item
		If $i_n=h$ the behavior of the simulator will be specified for each considered proposition (Proposition~\ref{prop.h1}, \ref{prop.h2}, \ref{prop.h3}, \ref{prop.h4}, \ref{prop.h5}).
		\item
		If $i_n>h$ we output a standard key.
	\end{itemize}
	
	\phase{Challenge}
	We receive the attributes $S_*$, the time $T_*$, and the messages $M_*^0$ and $M_*^1$ which the adversary wish to challenge, and then we fix a bit $b$ randomly in $\{0,1\}$.
	We want to create a semi-functional ciphertext without knowing the generator~$g_2$.
	To do so we fix implicitly $X_1=g^s$ and $Y_1=g_2^c$.
	
	The components of the ABE ciphertext are:
		\begin{align*}
			C_0&=X_1 Y_1;\\
			C_{1,i}&=(X_1 Y_1)^{t_i};\\
			\mathrm{CH}_\textnormal{ABE}&=\left(C_0,\{C_{1,i}\}_{i\in S}\right).
		\end{align*}
		
	We construct also the SUE ciphertext $\mathrm{CH}_\textnormal{SUE}$ using the element $X_1Y_1$.
	Lastly we observe that $e(X_1Y_1,g)=e(X_1,g)$, which is implicitly $e(g,g)^s$.
	The ciphertext we output to the adversary $\mathcal{A}$ is:
	\[
		\mathrm{CT}=\left(\mathrm{CH}_\textnormal{ABE},\mathrm{CH}_\textnormal{SUE},e(X_1Y_1,g)^\alpha \cdot M_*^b\right).
	\]
	
	The described constructions account for a correctly distributed ciphertext. Moreover, the ciphertext and the keys for the case $i_n\neq h$ are uncorrelated. Thus in the remaining propositions we must only consider the case $i_n=h$.
		
	\phase{Query 2}
	The second quering phase is managed as the first one.
	
	\phase{Guess}
	We obtain the guess $b$ from $\mathcal{A}$, which is the final output of our simulator.
	
	\vspace{0.5em} 
	
	To conclude the proof of Lemma~\ref{lemma.g1h} we need to prove the following propositions.
	In each of them we describe the behavior of the simulator when the counter index is $h$. We consider only the case where the aversary is of \typeone{}, for the same result for an adversary of \typetwo{} we refer to Lee et al.~\cite{lee2013RSABE}.
	
	\begin{proposition}
		\label{prop.h1}
		If Assumption \ref{a2} is true, then no PPT adversary $\mathcal{A}$ of \typeone{} is able to distinguish between \game{H}{k,1} and \game{H}{k-1,2}, that is, $\left| \adv_{\mathcal{A}}^{\mathgame{H}{k,1}} - \adv_{\mathcal{A}}^{\mathgame{H}{k-1,2}}\right|$ is negligible with respect to $\lambda$.
	\end{proposition}
	\begin{proof}
		If the query is for a SUE key, we give to the adversary a standard SUE key, which we are able to build. If the request is for an ABE key, we must consider three different cases depending on the counter index. Suppose that the access structure is described by $\mathbb{A}=(B,\rho)$, where $B\in\matrixset{\varRow}{\varColumn}{\Z_N}$ and $\rho$ is injective. The request is associated with the node $\nu_j$, corresponding to the node index $i_n=h$, which is associated with the element $\gamma_j$ in $\tree$.
		\begin{itemize}
			\item
			If $i_c<k$ we create a semi-functional ABE key of type 2.
			\item
			If $i_c=k$ we build the key relying on the element $W$.
			\randomchoose{p}{%
				$Z_{1,i},Z_{2,i}\in\G_{p_3}\text{, for all }i\in \{1,\ldots,\varRow\}$;
				$r_{2},\ldots,r_{\varColumn},r^\prime_{2},\ldots,r^\prime_{\varColumn}\in \Z_N$;
				$\delta_i\in\Z_N\text{, for all }i\in \{1,\ldots,\varRow\}$.
			}
			We set $\vec{t}=(0,r_{2},\ldots,r_{\varColumn})$ and $\pvec{t}=(\gamma_j,r^\prime_{2},\ldots,r^\prime_{\varColumn})$.			
			The given secret key is:			
			\[
				\mathrm{SK}_\textnormal{ABE}=\left( \left\{ g^{B_i\cdot \pvec{t}} W^{B_i\cdot\vec{t}}W^{\delta_i t_{\rho(i)}} Z_{1,i}, W^{\delta_i}Z_{2,i} \right\}_{i=1}^{\varRow} \right).
			\]
			\item
			If $i_c>k$ we create a standard ABE key.
		\end{itemize}
		The only case we need to check is when $i_c=k$.
		We call $g^r$ the $p_1$ part of $W$, i.e., $X_2=g^r$, and implicitly we set $g_2=Y_3$, the $p_2$ part of $W$, when it exists.
		The $p_3$ part of $W$ does not add any contribution, since each element composing the key is multiplied by a random element in $\G_{p_3}$.
		The $p_1$ part manages the standard key creation. In particular the private key we are building sets implicitly the vector $\vec{v}$ used in the function \algorithm{ABE}{GenKey} as $\vec{v}=r\vec{t}+\pvec{t}$, and the $s_i$ elements used in the same function implicitly as $s_i=r\delta_i$.
		If the $p_2$ part of $W$ is $1$, i.e., $W\in\G_{p_1p_3}$, then the output key is a standard key, as needed for game \game{H}{k-1,2}.
		Otherwise the $p_2$ part of $W$ is nontrivial, and we obtain a semi-functional key of type 1. In particular (we use the notation found in the definition of \game{H}{k,1} at the beginning of Lemma~\ref{lemma.g1h}) the elements $z_i$'s are equal to the $t_i$'s, the $\delta_i$'s are the same as in the construction of an ABE semi-functional key, and $\vec{u}=\vec{t}$.
		
		It remains to show that the given secret key is correctly distributed.
		First we can use the same argument as in Lemma~\ref{lemma.g0} to observe that the output distribution does not change because $t_i=z_i$. In fact we have previously chosen $t_i$ to be a random element of $\Z_N$, and the occurrence of $t_i$ uses its $p_2$ part if and only if they are used instead of $z_i$.
		
		The only remaining problem is that the first component of the vector $\vec{u}$ is not in general zero, but it is a random number in $\Z_N$. Nevertheless the final distribution remains the same in this case. 
		Since the set $S$ is not authorized, we know that the span of the rows associated with an attribute in~$S$ does not contain $(1,0,\ldots,0)$, otherwise $S$ would be able to decrypt the ciphertext.
		In particular the same is true if we consider the projection of each entry of the row from $\Z_N$ to $\Z_{p_2}$, and their corresponding linear span~$U$. Otherwise if we consider the same linear combination in $\Z_N$ then the first element is one plus a multiple of $p_2$, and the other elements are multiples of $p_2$: by computing the greatest common divisor between any multiple of $p_2$ and $N$ we are able to obtain a nontrivial factor of $N$, thus allowing us to break Assumption~\ref{a2}.
		Therefore we can find a vector $\vec{w}$ which is orthogonal to the space~$U$, but which is not orthogonal to $(1,0,\ldots,0)$. In fact if each element of $U^\perp$ is orthogonal to $(1,0,\ldots,0)$, then $U^\perp \subseteq\langle e_1\rangle^\perp$, and since the standard bilinear product is symmetric, non-degenerate, and the space has finite dimension, then $U \supseteq\langle e_1\rangle$, which is a contradiction. We can fix a basis including $\vec{w}$: this allows us to write $\vec{t}=f\vec{w}+\pvec{w}$ for some $f\in\Z_N$, where $\pvec{w}$ belongs to the span of the remaining elements of the base. We can choose $\pvec{w}$ with uniform distribution, and then we can fix the unique $f$ which zeroes the first component.
		 With this construction we obtain each possible element for $\vec{t}$, thus a uniform distribution of the $\pvec{v}$ in the semi-functional key corresponds to a uniform distribution for the choice of $\vec{t}$. 
		
		We want to show that an adversary has not enough information to discover~$f$: by knowing just $\pvec{w}$ we are not able to obtain any information about $f$.
		The only expression where $\vec{t}$ appears is of the form $B_i \vec{t} +\delta_i z_{\rho(i)}$.
		If we are considering a row $i$ for an attribute which belongs to the challenge set of attributes~$S$, then we are not giving any information about $f$, since $B_i\cdot \vec{w}=0$ by definition of $\vec{w}$.
		For the remaining attributes we use the hyphothesis for which $\rho$ is injective. In particular, there is only one occurrence of the element~$z_{\rho(i)}$, since it does not appear in the challenge ciphertext.
		The element $z_{\rho(i)}$ is randomly chosen, and unless $\delta_i=0$ we are not adding any information about $f$, since we are adding to the expression an unknown random element.
		Furthermore the probability that $\delta_i=0$ becomes negligible as $N$ grows.
		In particular, an adversary is not able to distinguish the distribution of the output key from a key created with random first component for $\vec{t}$.
		
		We have shown that the key given to the adversary in our simulator has the correct distribution. Therefore, we have shown that this is a correct simulator. Its existence would contradict our assumption, 
	\end{proof}
	\begin{proposition}
		\label{prop.h2}
		If Assumption \ref{a2} is true, then no PPT adversary $\mathcal{A}$ of \typeone{} is able to distinguish between \game{H}{k,1} and \game{H}{k,2}, meaning that $\left| \adv_{\mathcal{A}}^{\mathgame{H}{k,1}} - \adv_{\mathcal{A}}^{\mathgame{H}{k,2}}\right|$ is negligible with respect to $\lambda$.
	\end{proposition}
	\begin{proof}
		We consider the case $i_n=h$.
		As before, we consider a request for the access structure $\mathbb{A}=(B,\rho)$, where $B\in\matrixset{\varRow}{\varColumn}{\Z_N}$. The request is associated with the node $\nu_j$, which is associated with the element $\gamma_j$ in $\tree$.
		\begin{itemize}
			\item
			If $i_c<k$ we create a semi-functional ABE key of type 2.
			\item
			If $i_c=k$ we build the key relying on the element $W$.
			\randomchoose{p}{%
				$Z_{1,i},Z_{2,i}\in\G_{p_3}\text{, for all }i\in \{1,\ldots,\varRow\}$;
				$r_{1},\ldots,r_{\varColumn},r^\prime_{2},\ldots,r^\prime_{\varColumn}\in \Z_N$;
				$\delta_i\in\Z_N\text{, for all }i\in \{1,\ldots,\varRow\}$.
			}
			We set $\vec{t}=(r_1,r_{2},\ldots,r_{\varColumn})$ and $\pvec{t}=(\gamma_j,r^\prime_{2},\ldots,r^\prime_{\varColumn})$.
			
			The given secret key is:
			\[
				\mathrm{SK}_\textnormal{ABE}=\left( \left\{ g^{B_i\cdot \pvec{t}} (Y_2 Z_1)^{B_i\cdot\vec{t}}W^{\delta_i t_{\rho(i)}} Z_{1,i}, W^{\delta_i}Z_{2,i} \right\}_{i=1}^{\varRow} \right).
			\]
			\item
			If $i_c>k$ we create a standard ABE key.
		\end{itemize}
		The only case we need to check is when $i_c=k$.
		We call $g^r$ the $p_1$ part of $W$, i.e., $X_2=g^r$, and we set implicitly $Y_2=g_2$.
		Observe that the $p_1$ part of the key creates a standard key, which sets implicitly in the construction of a standard ABE key done in \algorithm{ABE}{GenKey} the element $s_i$ as $r \delta_i$, which is uniformly distributed in $\Z_N$.
		If $W$ contains no $p_2$ part, then the key is a semi-functional ABE key of type 2, and $\vec{v}=\vec{t}$.
		Otherwise if $g_2^d$ is the $p_2$ part of $W$, we have also the additional component needed for the key of type 1, with the same $\delta_i$'s and $z_i=t_i$.
		As before we notice how the reuse of $t_i$ as $z_i$ does not change the result, since we are considering in one case its $p_1$ part, and for the other variable its $p_2$ part.
		Moreover, the $p_3$ part contribution is absorbed by the random elements, and the other elements are independent.
		
		We proved that the simulator is correctly distributed, thus an adversary who could distinguish between the two games would also be able to break Assumption~\ref{a2}.
	\end{proof}
	\begin{proposition}
		\label{prop.h3}
		If Assumption \ref{a2} is true, then no PPT adversary $\mathcal{A}$ of \typeone{} is able to distinguish between \game{H^\prime}{k,1} and \game{H^\prime}{k-1,2}, that is, $\left| \adv_{\mathcal{A}}^{\mathgame{H^\prime}{k,1}} - \adv_{\mathcal{A}}^{\mathgame{H^\prime}{k-1,2}}\right|$ is negligible with respect to $\lambda$.
	\end{proposition}
	\begin{proof}
		This proof follows the same procedure as the Proposition~\ref{prop.h1}.
		In particular, each case is equal, with two exceptions, where we add the contribution of the factor $(Y_2Z_1)^{\zeta_j}$.
		\begin{itemize}
			\item
			If $i_c<k$ we follow the same procedure as in Proposition~\ref{prop.h1}.
			\item
			If $i_c=k$ we build the key relying on the element $W$.
			\randomchoose{p}{%
				$Z_{1,i},Z_{2,i}\in\G_{p_3}\text{, for all }i\in \{1,\ldots,\varRow\}$;
				$r_{2},\ldots,r_{\varColumn},r^\prime_{2},\ldots,r^\prime_{\varColumn},r^{\prime\prime}_{2},\ldots,\allowbreak{}r^{\prime\prime}_{\varColumn}\in \Z_N$;
				$\delta_i\in\Z_N\text{, for all }i\in \{1,\ldots,\varRow\}$.
			}
			We set $\vec{t}=(0,r_{2},\ldots,r_{\varColumn})$, $\pvec{t}=(\gamma_j,r^\prime_{2},\ldots,r^\prime_{\varColumn})$ and $\ppvec{t}=(\zeta_j,r^{\prime\prime}_{2},\ldots,r^{\prime\prime}_{\varColumn})$.
			The given secret key is:
			
			\[
				\mathrm{SK}_\textnormal{ABE}=\left( \left\{ g^{B_i\cdot \pvec{t}} W^{B_i\cdot\vec{t}}W^{\delta_i t_{\rho(i)}} Z_{1,i}(Y_2Z_1)^{B_i\cdot\ppvec{t}}, W^{\delta_i}Z_{2,i} \right\}_{i=1}^{\varRow} \right).
			\]
			\item
			If $i_c>k$ we create a semi-functional ABE key.
		\end{itemize}
		The rest of the proof is equal to the Proposition~\ref{prop.h1}, since the added factor does not change the output distribution.
	\end{proof}
	\begin{proposition}
		\label{prop.h4}
		If Assumption \ref{a2} is true, then no PPT adversary $\mathcal{A}$ of \typeone{} is able to distinguish between \game{H^\prime}{k,1} and \game{H^\prime}{k,2}, meaning that $\left| \adv_{\mathcal{A}}^{\mathgame{H^\prime}{k,1}} - \adv_{\mathcal{A}}^{\mathgame{H^\prime}{k,2}}\right|$ is negligible with respect to $\lambda$.
	\end{proposition}
	\begin{proof}
		This proof follows the same procedure as the Proposition~\ref{prop.h2} and Proposition~\ref{prop.h3}.
		In particular, each case is equal, with two exceptions, where we add the contribution of the factor $(Y_2Z_1)^{\zeta_j}$.
		\begin{itemize}
			\item
			If $i_c<k$ we follow the same procedure as in Proposition~\ref{prop.h2}.
			\item
			If $i_c=k$ we build the key relying on the element $W$.
			\randomchoose{p}{%
				$Z_{1,i},Z_{2,i}\in\G_{p_3}\text{, for all }i\in \{1,\ldots,\varRow\}$;
				$r_{1},\ldots,r_{\varColumn},r^\prime_{2},\ldots,\allowbreak{}r^\prime_{\varColumn},r^{\prime\prime}_{2},\ldots,r^{\prime\prime}_{\varColumn}\in \Z_N$;
				$\delta_i\in\Z_N\text{, for all }i\in \{1,\ldots,\varRow\}$.
			}
			We set $\vec{t}=(r_1,r_{2},\ldots,\allowbreak{}r_{\varColumn})$, $\pvec{t}=(\gamma_j,r^\prime_{2},\ldots,r^\prime_{\varColumn})$ and $\ppvec{t}=(\zeta_j,r^{\prime\prime}_{2},\ldots,r^{\prime\prime}_{\varColumn})$.
			
			The given secret key is:
			\[
				\mathrm{SK}_\textnormal{ABE}=\left( \left\{ g^{B_i\cdot \pvec{t}} (Y_2 Z_1)^{B_i\cdot\vec{t}}W^{\delta_i t_{\rho(i)}} Z_{1,i}(Y_2Z_1)^{B_i\cdot\ppvec{t}}, W^{\delta_i}Z_{2,i} \right\}_{i=1}^{\varRow} \right).
			\]
			\item
			If $i_c>k$ we create a semi-functional ABE key.
		\end{itemize}
		The rest of the proof is equal to the Proposition~\ref{prop.h2}, since the added factor does not change the output distribution.
	\end{proof}
	\begin{proposition}
		\label{prop.h5}
		If Assumption \ref{a2} is true, then no PPT adversary $\mathcal{A}$ of \typeone{} is able to distinguish between \game{H^\prime}{0,2} and \game{H^{\prime\prime}}{}, which means that $\left| \adv_{\mathcal{A}}^{\mathgame{H^\prime}{0,2}} - \adv_{\mathcal{A}}^{\mathgame{H^{\prime\prime}}{}}\right|$ is negligible with respect to $\lambda$.
	\end{proposition}
	\begin{proof}
		This proof is similar to Proposition~\ref{prop.h1}.
		Suppose that the adversary queries for an ABE key with node index $i_n=h$, for an access structure described by $\mathbb{A}=(B,\rho)$, where $B\in\matrixset{\varRow}{\varColumn}{\Z_N}$ and $\rho$ is injective. The request is associated with the node $\nu_j$. Notice that we do not need to distinguish between the different counter indices, since all keys for the two considered games with the same node index are equal.
		
		\randomchoose{p}{%
			$Z_{1,i},Z_{2,i}\in\G_{p_3}\text{, for all }i\in \{1,\ldots,\varRow\}$;
			\linebreak$r_{2},\ldots,r_{\varColumn},r^\prime_{2},\ldots,r^\prime_{\varColumn},s_1,\ldots,s_\varRow\in \Z_N$;
			$\zeta^\prime_j\in \Z_N$.
		}		
		We call $\vec{t}=(1,r_2,\ldots,r_\varColumn)$ and $\pvec{t}=(\zeta^\prime_j,r^\prime_2,\ldots,r^\prime_\varColumn)$.		
		The semi-functional ABE key is:
		\[
			\mathrm{SK}_{\textnormal{ABE}}=
			\left(\left\{
			W^{B_i\cdot \vec{t}}T_{\rho(i)}^{s_i}(Y_2Z_1)^{B_i\cdot \pvec{t}}Z_{1,i},
			g^{s_i}Z_{2,i}\right
			\}_{i=1}^{\varRow}\right).
		\]		
		\disclaimerSUE{}
		We compute explicitly the components of the ABE key where $W$ appears, ignoring the $p_3$ parts.
		\begin{align*}
			W^{B_i\cdot \vec{t}}T_{\rho(i)}^{s_i}(Y_2)^{B_i\cdot \pvec{t}}
			&=(g^{\gamma_j} g_2^{\delta})^{B_i\cdot \vec{t}}T_{\rho(i)}^{s_i}(g_2)^{B_i\cdot \pvec{t}}\\
			&=g^{\gamma_j B_i\cdot \vec{t}} T_{\rho(i)}^{s_i} g_2^{\delta B_i\cdot \vec{t}+B_i\cdot \pvec{t}}.
		\end{align*}
		The first component of the vector $\gamma_j B_i\cdot \pvec{t}$ is exactly $\gamma_j$, and the remaining are randomly distributed. If the $p_2$ part of $W$ is equal to $1$, i.e., $\delta=0$, then we have a standard semi-functional ABE key construction with $\zeta_j=\zeta^\prime_j$. Otherwise the first component of the vector $\delta B_i\cdot \vec{t}+B_i\cdot \pvec{t}$ is $\delta + \zeta^\prime_j$, which is $\zeta_j$. The remaining vector elements are again randomly distributed and independent on other key value, since from the second element onward of $\pvec{t}$ they are randomly fixed in~$\Z_N$.
		
	\end{proof}
	\subsubsection{Randomizing the ciphertext}
	\begin{lemma}
	\label{lemma.g2}
	If Assumption \ref{a3} is true, then no PPT adversary $\mathcal{A}$ is able to distinguish between \game{G}{2} and \game{G}{3}, which means that $\left| \adv_{\mathcal{A}}^{\mathgame{G}{3}} - \adv_{\mathcal{A}}^{\mathgame{G}{2}}\right|
	$ is negligible with respect to $\lambda$.
	\end{lemma}
	\begin{proof}
		By contradiction we prove that the existence of a PPT adversary who can distinguish between \game{G}{2} and \game{G}{3} implies the existence of an algorithm~$\mathscr{B}_3$ which can break Assumption~\ref{a3}.
		
		The algorithm $\mathscr{B}_3$ is given the two inputs $D=\left(\mathscr{S},g_1,g_2,g_3,g_1^\alpha Y_1,g_1^s Y_2\right)$ and $W$.
		$W$ can be $W_1=e(g_1,g_1)^{s\alpha}$ or a random element $W_2$ in $\G_{p_1}$, or equivalently $W_2=e(g_1,g_1)^c$ for some random $c\in \Z_N$.
		
	\phase{Setup}
	The setup phase is similar to the one in Lemma \ref{lemma.g0}, with the exception of the choice of the public parameter $\Omega$, for which we will (implicitly) use the same $\alpha$  used in the input.
	\randomchoose{p}{%
		$a,\gamma\in\Z_N$;
		$t_j^\prime \in \Z_N\text{, for all }j\in \mathfrak{A}$.
	}
	
	Let $g=g_1$. We set the ABE keys as:
	\begin{align*}
		\mathrm{MK}_{\textnormal{ABE}}&=\left( \gamma, Z=g_3 \right);\\
		\mathrm{PK}_{\textnormal{ABE}}&=\left( \left(N,\G,\GT,e\right), g, \left\{g^{t_j}\right\}_{j\in \mathfrak{A}}, e(g,g)^\gamma \right).
	\end{align*}
	\disclaimerSUE{}
	Since we are choosing random exponents, when we consider their corresponding elements of the group, we are choosing them with a uniform distribution. Thus the output of this procedure has the same probability distribution of the corresponding outputs of the standard setup procedure.
	We create the binary tree $\tree$ associating to each node $\nu_i$ a random exponent~$\gamma_{i}$.
	Lastly we set $\Omega=e(g,g^\alpha Y_1)$, which for the bilinearity is equivalent to $e(g,g)^\alpha$. In particular we are implicitly setting $
		\mathrm{MK}=\left(\alpha, \tree, \mathrm{MK}_{\textnormal{ABE}}, \mathrm{MK}_{\textnormal{SUE}} \right)$.
	The public key, given to $\mathcal{A}$, is then:
	\begin{align*}
		\mathrm{PK}&=\left( \mathrm{PK}_{\textnormal{ABE}}, \mathrm{PK}_{\textnormal{SUE}}, g, \Omega \right).
	\end{align*}
	We also generate randomly for each node $\nu_i$ the elements $\zeta_i$, $\eta^\prime_i$ in $\Z_N$, which are needed to create the semi-functional keys.
	
	\phase{Query 1}
	In this phase we supply $\mathcal{A}$ with all the required keys, both private keys and time-update keys. 
	Since we do not know $\alpha$, we are not able to create standard time-update keys; nevertheless we will show that we can still create semi-functional keys from our inputs. We fix $g_2$ as the generator of $\G_{p_2}$ used for every semi-functional request.
	\begin{itemize}
		\item
		If the request is for a private key, then we first find the private set $\mathrm{PV}_u$ for the user $u$. For each node $\nu_h$ involved with the private set, we create its corresponding ABE key. This ABE key should be generated with an ABE scheme of private master key $\mathrm{MK}_{\textnormal{ABE}}, \mathrm{PK}_{\textnormal{SUE}}=(\gamma_h,g_3)$: since all elements are known, we are able to generate the corresponding ABE private key, and then we can apply the same procedure in \algorithm{RS-ABE}{GenKeySF} to make this key semi-functional .
		\item
		Otherwise the request is for a time-update key, for which we use the procedure shown in Lee et al.~\cite{lee2013RSABE}.
		
	\end{itemize}
	
	\phase{Challenge}
	We receive the attributes $S_*$, the time $T_*$, and the messages $M_*^0$ and $M_*^1$ upon which the adversary wishes to be challenged; we fix a bit $b$ randomly in $\{0,1\}$.
	We use both these and the inputs of $\mathscr{B}_3$ in order to create the needed semi-functional ciphertext, using randomly chosen exponents.
	First we generate the ABE semi-functional ciphertext.
	The components of the ciphertext are:
		\begin{align*}
			C_0=g^s Y_2;\ 
			C_{1,i}=(g_1^s Y_2)^{t_i};\ 
			\mathrm{CH}_\textnormal{ABE}=\left(C_0,\{C_{1,i}\}_{i\in S}\right).
		\end{align*}
	We generate also the SUE ciphertext $\mathrm{CH}_\textnormal{SUE}$, as done by Lee et al.~\cite{lee2013RSABE}.
	
	Finally we output to the adversary the ciphertext
	\[
		\mathrm{CT}=\left(\mathrm{CH}_\textnormal{ABE},\mathrm{CH}_\textnormal{SUE},W\cdot M_*^b\right).
	\]
	
	We need to prove that the above definition gives rise to the games \game{G}{2} and \game{G}{3} depending on the form of the given value $W$. In particular we must show that the above construction generates valid semi-functional ciphertexts. Notice that the construction of the keys uses the term $g^\alpha Y_1$, and the ciphertexts contruction uses $g^s Y_2$, elements which are chosen independently. Thus there is not mutual correlation between them.
	
	The constructed ABE ciphertext is equivalent to a semi-functional ABE ciphertext with $c=t$. The construction of an ABE semi-functional ciphertext involves the choice of some random elements $z_i$ for each used attribute, but we have implicitly set $z_i=t_i$, values used in the previous construction. Nevertheless this setting is independent from the previous random choices, since this is the only time we have considered the $p_2$ part of the $t_i$'s.
	Previously we have only used them in the form $g^{t_i}$, which uses only the $p_1$ part of the $t_i$'s. Using the same argument of Lemma~\ref{lemma.g0}, the Chinese remainder theorem assures us that the two parts are mutually independent.
	
	If the element $W$ has the form $W_1=e(g_1,g_1)^{s\alpha}$, then the output is exactly a semi-functional ciphertext, as required by the game \game{G}{2}, otherwise it is of the form $W_2=e(g_1,g_1)^c$ for some random $c\in \Z_N$, which is the output required for the game \game{G}{3}.
	
	\phase{Query 2}
	We repeat the keys construction of the first querying phase.
	
	\phase{Guess}
	We obtain the guess $b$ from $\mathcal{A}$, which is the final output of $\mathscr{B}_3$.
	
	The adversary $\mathcal{A}$ has a non-negligible advantage on differentiating \game{G}{2} from \game{G}{3}, which means that we can correctly guess the type of $W$: the algorithm $\mathscr{B}_3$ has a non-negligible advantage in breaking Assumption \ref{a3}, which is a contradiction.
	\end{proof}
	
	
	\section{Efficiency and conclusions}
	\label{efficiency.sec}
	In this section we study the dimension of keys and ciphertexts of RS-ABE.
	The dimension is measured in term of group elements which compose the keys and the ciphertext, and depends on the parameters of the scheme.
	
	For the SUE scheme we refer to~\cite{lee2013RSABE}. The maximum length for a SUE key is $\left\lceil \log_2(T_\textnormal{max}+2) \right\rceil+1$, where $T_\textnormal{max}$ is the maximum time for the system, and a SUE cipertext header has size that is upper-bounded by $(d+2)+2(d)=3d+2\leq 3 \left\lceil \log_2(T_\textnormal{max}+2) \right\rceil -1$.
	
	We analyze now the ABE scheme, starting from the key. We want to find the length of a key associated with an LSSS access structure $(B,\rho)$. The first thing to notice is the hypothesis stated at the beginning of Section~\ref{ABE.construction}, requiring that $\rho$ is injective, i.e., each attribute is assigned with a single row of the matrix $B$. We can extend the result by assuming that $\rho$ is associated with at most $k$ rows of the matrix $B$. We modify the attributes used in the scheme: if $\mathfrak{A}$ is the set of attributes of the scheme, we consider the set $\mathfrak{A}^\prime=\{(u,i)\mid u\in\mathfrak{A},i\in\{1,\ldots,k\}\}$ which contains $k$ copies of each attribute. Then each attribute set $S$ becomes the set $S^\prime=\{(u,i)\mid u\in S,i\in\{1,\ldots,k\}\}$, and the new function $\rho$ assigns to each occurrence of an attribute one of the new corresponding attributes, without repetitions.
	
	An ABE private key for an LSSS access structure $(B,\rho)$, modified to take into account for repetitions in $\rho$, such that $B\in \matrixset{\varRow}{\varColumn}{\Z_N}$ is a matrix with $\varRow$ rows and $\varColumn$ columns, has the following form:
	\[
		\mathrm{SK}_{\mathbb{A}}=\left(\left\{K_{1,i}=g^{B_i\cdot \vec{v}} T_{\rho(i)}^{s_i} Z_{1,i}, K_{2,i}=g^{s_i} Z_{2,i}\right\}_{1\leq i\leq \varRow} \right).
	\]
	A key contains exactly $2l$ group elements.
	
	An ABE ciphertext for the set of attributes $S$ has the form\[
		\mathrm{CH}_{S}=\left(C_0=g^s,\left\{C_{1,i}=T_i^s\right\}_{i\in S} \right),
	\]
	which corresponds to a size equal to $|S|+1$.
	
	Now we merge these results together to obtain the length of the keys for the RS-ABE scheme.
	The private key for a user $u$ and an access structure $\mathbb{A}$ is generated starting from the private set $\mathrm{PV}_u$ associated with the user $u$: a private ABE key for the access structure $\mathbb{A}$ is generated for each element of $\mathrm{PV}_u$ . Then the private key is: \[
		\mathrm{SK}_{\mathbb{A},u}=\left( \mathrm{PV}_u, \mathrm{SK}_{\textnormal{ABE},1}, \ldots, \mathrm{SK}_{\textnormal{ABE},d} \right).
	\]
	The number of elements $d$ of the private set is smaller than $\log_2(N_\textnormal{max})$, as we have noticed in Section~\ref{cs.sec}. This means that the maximum length of a private RS-ABE key, considering only group elements, is bounded by $2 l N_\textnormal{max}$.
	
	A time-update key can be obtained starting from a covering of the associated revoked set $R$: for each element of the covering $\mathrm{CV}_R$ a SUE key is built. Calling $r=|R|$, we know from Section~\ref{cs.sec} that the size of the covering is bounded above by the value $r \log_2 (N_\textnormal{max}/r)$. Therefore the maximum number of group elements is	$3 r \log_2 \left(\frac{N_\textnormal{max}}{r}\right)  \left\lceil \log_2(T_\textnormal{max}+2) \right\rceil$.
	
	A ciphertext for a set of attributes $S$ is composed by an element of $\GT$ and both an ABE and a SUE ciphertext. The number of group elements is then $|S|+\left\lceil\log_2(T_\textnormal{max}+2) \right\rceil+1$.
	\begin{table}[H]
		\begin{center}
		\makebox[\textwidth][c]{ 
		\renewcommand{\arraystretch}{1.2}

		\begin{tabular}{M{3.9cm}| M{4.55cm}| M{4.1cm} !{\color{white}} N } 
				&	{\textbf{Upper bound}}	&	{\textbf{Magnitude}}	&	\\
			\hline 
			SUE private key	&	$\left\lceil \log_2(T_\textnormal{max}+2) \right\rceil+1$	&	$\bigO{\log_2T_\textnormal{max}}$ \\
			SUE ciphertext	&	$3 \left\lceil \log_2(T_\textnormal{max}+2) \right\rceil -1$	&	$\bigO{\log_2T_\textnormal{max}}$ \\
			ABE	private key	&	$2l$	&	$\bigO{l}$ \\
			ABE ciphertext	&	$|S|+1$	&	$\bigO{|S|}$ \\
			RS-ABE private key	&	$2 l N_\textnormal{max}$	&	$\bigO{l N_\textnormal{max}}$ \\
			RS-ABE time-update key	&	${3 r\! \log_2 \left(\frac{N_\textnormal{max}}{r}\right)\! \left\lceil \log_2(T_\textnormal{max}\!+\!2) \right\rceil}$	&	${\!\mathcal{O}(r \!\log_2 (N_\textnormal{max}/r)\!\log_2\! T_\textnormal{max})}$ \\ 
			RS-ABE ciphertext	&	$|S|+\left\lceil\log_2(T_\textnormal{max}+2) \right\rceil+1$	&	$\bigO{|S|+\log_2 T_\textnormal{max}}$ \\
			
		\end{tabular}
	}
	\end{center}
	\caption{The length in term of number of group elements employed. $T_\textnormal{max}$ is the maximum time used in the scheme, $l$ is the number of rows of the matrix in the access structure used for the key, $|S|$ is the size of the set of attributes used for the ciphertext, $N_\textnormal{max}$ is the maximum number of users, and $r$ is the size of the set of revoked users in the update key.}
	\end{table}

Finally we present briefly a possible extension of our work. In this paper we built a Key-Policy ABE scheme that allows to revoke users arbitrarily, and we prove its security. In our scheme only a single authority can create the public parameters and issue private keys to the users: this might be seen as a limitation for certain practical application of the scheme, for example those in which the users fear that the authority becomes curious. A natural follow-up for this work consists in introducing a hierarchy to manage the key-generation process, in a similar fashion to what is done for SSL certificates. A different problem is the trust on the authority managing the scheme: in our setting the authority has the full power to create and revoke the keys according to its own volition. This problem might be assuaged by distributing the authority's power among multiple authorities, such that a single failure would not cause a complete collapse of the scheme. In \cite{longocollaborative,longo2015key}, the authors present two different Key-Policy Attribute-Based Encryption schemes in which many different authorities operate independently, and prove their security. As a future work may be interesting the design of a Multi-Authority Key-Policy ABE scheme that allows user revocation, with a full proof of security.

	
	\section*{Acknowledgements}
	Most of the results shown in this work were developed in the first author's Master's thesis at the University of Padua, who would like to thank his supervisors: the third author and A.~Tonolo.


	\printbibliography{}
\end{document}